\documentclass{article}

\usepackage{amsmath}
\usepackage{amsthm}
\usepackage{amssymb}
\usepackage{url}
\usepackage{mathtools}
\usepackage{stmaryrd}
\usepackage{cleveref}
\usepackage{enumerate}
\usepackage[all]{xy}
\usepackage{xspace}
\usepackage{mdframed}
\usepackage{wrapfig}
\usepackage{color}

\usepackage{authblk}

\newtheorem{theorem}{Theorem}[section]
\crefname{theorem}{Theorem}{Theorems}

\newtheorem{corollary}[theorem]{Corollary}
\newtheorem{lemma}[theorem]{Lemma}
\newtheorem{proposition}[theorem]{Proposition}
\theoremstyle{definition}
\newtheorem{example}[theorem]{Example}
\newtheorem{definition}[theorem]{Definition}
\crefname{definition}{Definition}{Definitions}
\newtheorem{question}[theorem]{Question}

\newif\ifdraft\draftfalse

\ifdraft
\usepackage{todonotes}
\else
\usepackage[disable]{todonotes}
\renewcommand{\textcolor}[2]{#2}
\renewcommand{\paragraph}[1]{\vskip.4\baselineskip\noindent$\rhd$ \textit{#1}$\;$}
\newenvironment{mdexample}{
\begin{example}}
{
\end{example}}
\fi

\newif\ifarxiv\arxivtrue

\newcommand{\sref}[1]{Section \ref{#1}}

\newcommand{\pfgt}[1]{U_{#1}}

\newcommand{\mergefunc}{structure\xspace}
\newcommand{\duplicator}{\mathbf{D}}
\newcommand{\spoiler}{\mathbf{S}}

\newcommand{\lam}[1]{\lambda #1~.~}
\newcommand{\fa}[1]{\forall #1~.~}
\newcommand{\ex}[1]{\exists #1~.~}
\newcommand{\defeq}{\coloneq}
\newcommand{\fami}[3]{\textcolor{red}{\left({#1}\right)_{#2\in #3}}}
\newcommand{\vfami}[3]{\textcolor{red}{\overrightarrow{{#1}}}}
\newcommand{\intv}
{\II}

\newcommand{\ttt}{\mathrm{true}}
\newcommand{\fff}{\mathrm{false}}

\newenvironment{choice}{\left\{\begin{array}{ll}}{\end{array}\right.}
\renewcommand{\angle}[1]{\langle #1\rangle}

\newcommand{\ol}[1]{\overline{#1}}
\newcommand{\bul}{\mathbin{\bullet}}
\newcommand{\Eq}[1][]{\textcolor{red}{\mathrm{Eq}_{#1}}}

\newcommand{\laxsl}[2]{\textcolor{red}{{#1}/\!/{#2}}}
\newcommand{\clat}{{\mathbf{CLat}_{\sqcap}}}

\newcommand{\Id}{\mathrm{Id}}

\newcommand{\PredNA}{\mathbf{Pred}}
\newcommand{\Set}{\mathbf{Set}}
\newcommand{\Pre}{\mathbf{Pre}}
\newcommand{\Top}{\mathbf{Top}}
\newcommand{\EqRel}{\mathbf{EqRel}}
\newcommand{\ERelNA}{\mathbf{ERel}}

\newcommand{\PMet}{\mathbf{PMet}}
\newcommand{\QPMet}{\mathbf{LMet}} \newcommand{\Coalg}[1]{\mathbf{Coalg}(#1)}
\newcommand{\UCoalg}[1]{U_{#1}}
\newcommand{\CAT}{\mathbf{CAT}}
\newcommand{\cod}{\mathrm{cod}}

\newcommand{\Fda}{\textcolor{red}{{F_{\mathrm{da}}}}}
\newcommand{\dFda}{\textcolor{red}{{\dot F_{\mathrm{da}}}}}

\newcommand{\decent}[1]{\dot{#1}}
\newcommand{\dto}{\mathbin{\decent\to}}
\newcommand{\dtimes}{\mathbin{\decent\times}}

\newcommand{\subdistmnd}{{\mathcal{D}}}
\newcommand{\pow}{{\mathcal{P}}}
\newcommand{\powsub}{{\pow\subdistmnd}}

\newcommand{\sigmap}{{\sigma_{\mathrm{av}}}}
\newcommand{\sigmat}{{\sigma_{\otimes}}}
\newcommand{\sigmant}{{\sigma_{\oplus}}}

\newcommand{\distpow}{\lambda^{\pow}}

\newcommand{\distFda}{\lambda^{\mathrm{da}}}

\newcommand{\coalg}[3]{\textcolor{red}{#2\colon #3\to #1(#3)}}

\newcommand{\funcdist}[2]{\textcolor{red}{{#1}_{#2}}}
\newcommand{\coalgmerge}[3]{\funcdist{#1}{#2}(#3)}

\newcommand{\Bisim}[2]{\mathbf{Bisim}(#1,#2)}

\newcommand{\cptpl}[2]
        {\angle{#1}_{#2}}

\newcommand{\Sp}[3][]{\mathrm{Sp}^{#1}(#2,#3)}
\newcommand{\Spf}[3][]{r^{#1}_{#2,#3}}
\newcommand{\SP}[1][]{\mathbf{Sp}^{#1}}

\newcommand{\abs}[2]{L^{#1,#2}}
\newcommand{\conc}[2]{R^{#1,#2}}

\newcommand{\das}{d_{\intv}^{\mathrm{as}}}

\newcommand{\pfib}[6]
{(#5,#6)\colon Id_{#4}\to(#1\colon #2\to #3)}
\newcommand{\pfibU}[6]{(#5,#6)\colon Id_{#4}\to \pfgt{#2}}
\newcommand{\ptfib}{\textcolor{black}{$\clat$-fibration with truth values\xspace}}
\newcommand{\ptfibs}{\textcolor{black}{$\clat$-fibrations with truth values\xspace}}

\newcommand{\codlift}[3]
{\textcolor{red}{[#1,#3]^{#2}}}

\newcommand{\Ncodlift}[4]
{\codlift{#1}{#2}{#3}}

\newcommand{\cgame}[1]{\mathcal{G}_{#1}} \newcommand{\compcgame}[5]{\mathcal{G}_{\vfami{#1}{#2}{#3}}^{#4, #5}}

\newcommand{\lifttimes}[1]{\textcolor{red}{\times^{#1}}}

\newcommand{\rloop}[2][-]{\save \POS!R(.7) \ar@(ru,rd)^#1{#2} \restore}
\newcommand{\lloop}[2][-]{\save \POS!L(.7) \ar@(lu,ld)_#1{#2} \restore}

\newcommand{\adjdashv}[2]{
  \ar@<#2>@{}[#1]|-*=0[@]\txt{\tiny{$\bot$}} }
\newcommand{\adjunction}[3]{
  \ar@<.3pc>[#1]^-{#2}
  \adjdashv{#1}{0pc}
\ar@<-.3pc>@{<-}[#1]_-{#3}
}

\newcommand{\pbcorner}[1][dr]{\save*!/#1+1.2pc/#1:(1,-1)@^{|-}\restore}

\newcommand{\dF}{{\dot F}}
\newcommand{\dG}{{\dot G}}
\newcommand{\dT}{{\dot T}}

\renewcommand{\O}{\Omega}
\renewcommand{\P}{\Pi}
\newcommand{\bO}{\mathbf{\Omega}}
\newcommand{\bP}{\mathbf{\Pi}}

\newcommand{\BB}{\mathbb B}
\newcommand{\CC}{\mathbb C}
\newcommand{\DD}{\mathbb D}
\newcommand{\EE}{\mathbb E}
\newcommand{\FF}{\mathbb F}

\newcommand{\II}{\mathbb I}

\newcommand{\rue}{\ar@{=}[u]}
\newcommand{\rueh}[1]{\ar@{=}[u]^-{#1}}
\newcommand{\ruem}[1]{\ar@{=}[u]_-{#1}}
\newcommand{\ruen}[1]{\ar@{=}[u]|-{#1}}

\newcommand{\ruue}{\ar@{=}[uu]}
\newcommand{\ruueh}[1]{\ar@{=}[uu]^-{#1}}
\newcommand{\ruuem}[1]{\ar@{=}[uu]_-{#1}}
\newcommand{\ruuen}[1]{\ar@{=}[uu]|-{#1}}

\newcommand{\ruuue}{\ar@{=}[uuu]}
\newcommand{\ruuueh}[1]{\ar@{=}[uuu]^-{#1}}
\newcommand{\ruuuem}[1]{\ar@{=}[uuu]_-{#1}}
\newcommand{\ruuuen}[1]{\ar@{=}[uuu]|-{#1}}

\newcommand{\rdh}[1]{\ar[d]^-{#1}}
\newcommand{\rdm}[1]{\ar[d]_-{#1}}

\newcommand{\rde}{\ar@{=}[d]}
\newcommand{\rdeh}[1]{\ar@{=}[d]^-{#1}}
\newcommand{\rdem}[1]{\ar@{=}[d]_-{#1}}
\newcommand{\rden}[1]{\ar@{=}[d]|-{#1}}

\newcommand{\rdde}{\ar@{=}[dd]}
\newcommand{\rddeh}[1]{\ar@{=}[dd]^-{#1}}
\newcommand{\rddem}[1]{\ar@{=}[dd]_-{#1}}
\newcommand{\rdden}[1]{\ar@{=}[dd]|-{#1}}

\newcommand{\rddde}{\ar@{=}[ddd]}
\newcommand{\rdddeh}[1]{\ar@{=}[ddd]^-{#1}}
\newcommand{\rdddem}[1]{\ar@{=}[ddd]_-{#1}}
\newcommand{\rddden}[1]{\ar@{=}[ddd]|-{#1}}

\newcommand{\rrh}[1]{\ar[r]^-{#1}}

\newcommand{\rre}{\ar@{=}[r]}
\newcommand{\rreh}[1]{\ar@{=}[r]^-{#1}}
\newcommand{\rrem}[1]{\ar@{=}[r]_-{#1}}
\newcommand{\rren}[1]{\ar@{=}[r]|-{#1}}

\newcommand{\rrue}{\ar@{=}[ru]}
\newcommand{\rrueh}[1]{\ar@{=}[ru]^-{#1}}
\newcommand{\rruem}[1]{\ar@{=}[ru]_-{#1}}
\newcommand{\rruen}[1]{\ar@{=}[ru]|-{#1}}

\newcommand{\rruue}{\ar@{=}[ruu]}
\newcommand{\rruueh}[1]{\ar@{=}[ruu]^-{#1}}
\newcommand{\rruuem}[1]{\ar@{=}[ruu]_-{#1}}
\newcommand{\rruuen}[1]{\ar@{=}[ruu]|-{#1}}

\newcommand{\rruuue}{\ar@{=}[ruuu]}
\newcommand{\rruuueh}[1]{\ar@{=}[ruuu]^-{#1}}
\newcommand{\rruuuem}[1]{\ar@{=}[ruuu]_-{#1}}
\newcommand{\rruuuen}[1]{\ar@{=}[ruuu]|-{#1}}

\newcommand{\rrde}{\ar@{=}[rd]}
\newcommand{\rrdeh}[1]{\ar@{=}[rd]^-{#1}}
\newcommand{\rrdem}[1]{\ar@{=}[rd]_-{#1}}
\newcommand{\rrden}[1]{\ar@{=}[rd]|-{#1}}

\newcommand{\rrdde}{\ar@{=}[rdd]}
\newcommand{\rrddeh}[1]{\ar@{=}[rdd]^-{#1}}
\newcommand{\rrddem}[1]{\ar@{=}[rdd]_-{#1}}
\newcommand{\rrdden}[1]{\ar@{=}[rdd]|-{#1}}

\newcommand{\rrddde}{\ar@{=}[rddd]}
\newcommand{\rrdddeh}[1]{\ar@{=}[rddd]^-{#1}}
\newcommand{\rrdddem}[1]{\ar@{=}[rddd]_-{#1}}
\newcommand{\rrddden}[1]{\ar@{=}[rddd]|-{#1}}

\newcommand{\rrrh}[1]{\ar[rr]^-{#1}}
\newcommand{\rrrm}[1]{\ar[rr]_-{#1}}

\newcommand{\rrre}{\ar@{=}[rr]}
\newcommand{\rrreh}[1]{\ar@{=}[rr]^-{#1}}
\newcommand{\rrrem}[1]{\ar@{=}[rr]_-{#1}}
\newcommand{\rrren}[1]{\ar@{=}[rr]|-{#1}}

\newcommand{\rrrue}{\ar@{=}[rru]}
\newcommand{\rrrueh}[1]{\ar@{=}[rru]^-{#1}}
\newcommand{\rrruem}[1]{\ar@{=}[rru]_-{#1}}
\newcommand{\rrruen}[1]{\ar@{=}[rru]|-{#1}}

\newcommand{\rrruue}{\ar@{=}[rruu]}
\newcommand{\rrruueh}[1]{\ar@{=}[rruu]^-{#1}}
\newcommand{\rrruuem}[1]{\ar@{=}[rruu]_-{#1}}
\newcommand{\rrruuen}[1]{\ar@{=}[rruu]|-{#1}}

\newcommand{\rrruuue}{\ar@{=}[rruuu]}
\newcommand{\rrruuueh}[1]{\ar@{=}[rruuu]^-{#1}}
\newcommand{\rrruuuem}[1]{\ar@{=}[rruuu]_-{#1}}
\newcommand{\rrruuuen}[1]{\ar@{=}[rruuu]|-{#1}}

\newcommand{\rrrde}{\ar@{=}[rrd]}
\newcommand{\rrrdeh}[1]{\ar@{=}[rrd]^-{#1}}
\newcommand{\rrrdem}[1]{\ar@{=}[rrd]_-{#1}}
\newcommand{\rrrden}[1]{\ar@{=}[rrd]|-{#1}}

\newcommand{\rrrdde}{\ar@{=}[rrdd]}
\newcommand{\rrrddeh}[1]{\ar@{=}[rrdd]^-{#1}}
\newcommand{\rrrddem}[1]{\ar@{=}[rrdd]_-{#1}}
\newcommand{\rrrdden}[1]{\ar@{=}[rrdd]|-{#1}}

\newcommand{\rrrddde}{\ar@{=}[rrddd]}
\newcommand{\rrrdddeh}[1]{\ar@{=}[rrddd]^-{#1}}
\newcommand{\rrrdddem}[1]{\ar@{=}[rrddd]_-{#1}}
\newcommand{\rrrddden}[1]{\ar@{=}[rrddd]|-{#1}}

\newcommand{\rrrre}{\ar@{=}[rrr]}
\newcommand{\rrrreh}[1]{\ar@{=}[rrr]^-{#1}}
\newcommand{\rrrrem}[1]{\ar@{=}[rrr]_-{#1}}
\newcommand{\rrrren}[1]{\ar@{=}[rrr]|-{#1}}

\newcommand{\rrrrue}{\ar@{=}[rrru]}
\newcommand{\rrrrueh}[1]{\ar@{=}[rrru]^-{#1}}
\newcommand{\rrrruem}[1]{\ar@{=}[rrru]_-{#1}}
\newcommand{\rrrruen}[1]{\ar@{=}[rrru]|-{#1}}

\newcommand{\rrrruue}{\ar@{=}[rrruu]}
\newcommand{\rrrruueh}[1]{\ar@{=}[rrruu]^-{#1}}
\newcommand{\rrrruuem}[1]{\ar@{=}[rrruu]_-{#1}}
\newcommand{\rrrruuen}[1]{\ar@{=}[rrruu]|-{#1}}

\newcommand{\rrrruuue}{\ar@{=}[rrruuu]}
\newcommand{\rrrruuueh}[1]{\ar@{=}[rrruuu]^-{#1}}
\newcommand{\rrrruuuem}[1]{\ar@{=}[rrruuu]_-{#1}}
\newcommand{\rrrruuuen}[1]{\ar@{=}[rrruuu]|-{#1}}

\newcommand{\rrrrde}{\ar@{=}[rrrd]}
\newcommand{\rrrrdeh}[1]{\ar@{=}[rrrd]^-{#1}}
\newcommand{\rrrrdem}[1]{\ar@{=}[rrrd]_-{#1}}
\newcommand{\rrrrden}[1]{\ar@{=}[rrrd]|-{#1}}

\newcommand{\rrrrdde}{\ar@{=}[rrrdd]}
\newcommand{\rrrrddeh}[1]{\ar@{=}[rrrdd]^-{#1}}
\newcommand{\rrrrddem}[1]{\ar@{=}[rrrdd]_-{#1}}
\newcommand{\rrrrdden}[1]{\ar@{=}[rrrdd]|-{#1}}

\newcommand{\rrrrddde}{\ar@{=}[rrrddd]}
\newcommand{\rrrrdddeh}[1]{\ar@{=}[rrrddd]^-{#1}}
\newcommand{\rrrrdddem}[1]{\ar@{=}[rrrddd]_-{#1}}
\newcommand{\rrrrddden}[1]{\ar@{=}[rrrddd]|-{#1}}

\newcommand{\rle}{\ar@{=}[l]}
\newcommand{\rleh}[1]{\ar@{=}[l]^-{#1}}
\newcommand{\rlem}[1]{\ar@{=}[l]_-{#1}}
\newcommand{\rlen}[1]{\ar@{=}[l]|-{#1}}

\newcommand{\rlue}{\ar@{=}[lu]}
\newcommand{\rlueh}[1]{\ar@{=}[lu]^-{#1}}
\newcommand{\rluem}[1]{\ar@{=}[lu]_-{#1}}
\newcommand{\rluen}[1]{\ar@{=}[lu]|-{#1}}

\newcommand{\rluue}{\ar@{=}[luu]}
\newcommand{\rluueh}[1]{\ar@{=}[luu]^-{#1}}
\newcommand{\rluuem}[1]{\ar@{=}[luu]_-{#1}}
\newcommand{\rluuen}[1]{\ar@{=}[luu]|-{#1}}

\newcommand{\rluuue}{\ar@{=}[luuu]}
\newcommand{\rluuueh}[1]{\ar@{=}[luuu]^-{#1}}
\newcommand{\rluuuem}[1]{\ar@{=}[luuu]_-{#1}}
\newcommand{\rluuuen}[1]{\ar@{=}[luuu]|-{#1}}

\newcommand{\rlde}{\ar@{=}[ld]}
\newcommand{\rldeh}[1]{\ar@{=}[ld]^-{#1}}
\newcommand{\rldem}[1]{\ar@{=}[ld]_-{#1}}
\newcommand{\rlden}[1]{\ar@{=}[ld]|-{#1}}

\newcommand{\rldde}{\ar@{=}[ldd]}
\newcommand{\rlddeh}[1]{\ar@{=}[ldd]^-{#1}}
\newcommand{\rlddem}[1]{\ar@{=}[ldd]_-{#1}}
\newcommand{\rldden}[1]{\ar@{=}[ldd]|-{#1}}

\newcommand{\rlddde}{\ar@{=}[lddd]}
\newcommand{\rldddeh}[1]{\ar@{=}[lddd]^-{#1}}
\newcommand{\rldddem}[1]{\ar@{=}[lddd]_-{#1}}
\newcommand{\rlddden}[1]{\ar@{=}[lddd]|-{#1}}

\newcommand{\rlle}{\ar@{=}[ll]}
\newcommand{\rlleh}[1]{\ar@{=}[ll]^-{#1}}
\newcommand{\rllem}[1]{\ar@{=}[ll]_-{#1}}
\newcommand{\rllen}[1]{\ar@{=}[ll]|-{#1}}

\newcommand{\rllue}{\ar@{=}[llu]}
\newcommand{\rllueh}[1]{\ar@{=}[llu]^-{#1}}
\newcommand{\rlluem}[1]{\ar@{=}[llu]_-{#1}}
\newcommand{\rlluen}[1]{\ar@{=}[llu]|-{#1}}

\newcommand{\rlluue}{\ar@{=}[lluu]}
\newcommand{\rlluueh}[1]{\ar@{=}[lluu]^-{#1}}
\newcommand{\rlluuem}[1]{\ar@{=}[lluu]_-{#1}}
\newcommand{\rlluuen}[1]{\ar@{=}[lluu]|-{#1}}

\newcommand{\rlluuue}{\ar@{=}[lluuu]}
\newcommand{\rlluuueh}[1]{\ar@{=}[lluuu]^-{#1}}
\newcommand{\rlluuuem}[1]{\ar@{=}[lluuu]_-{#1}}
\newcommand{\rlluuuen}[1]{\ar@{=}[lluuu]|-{#1}}

\newcommand{\rllde}{\ar@{=}[lld]}
\newcommand{\rlldeh}[1]{\ar@{=}[lld]^-{#1}}
\newcommand{\rlldem}[1]{\ar@{=}[lld]_-{#1}}
\newcommand{\rllden}[1]{\ar@{=}[lld]|-{#1}}

\newcommand{\rlldde}{\ar@{=}[lldd]}
\newcommand{\rllddeh}[1]{\ar@{=}[lldd]^-{#1}}
\newcommand{\rllddem}[1]{\ar@{=}[lldd]_-{#1}}
\newcommand{\rlldden}[1]{\ar@{=}[lldd]|-{#1}}

\newcommand{\rllddde}{\ar@{=}[llddd]}
\newcommand{\rlldddeh}[1]{\ar@{=}[llddd]^-{#1}}
\newcommand{\rlldddem}[1]{\ar@{=}[llddd]_-{#1}}
\newcommand{\rllddden}[1]{\ar@{=}[llddd]|-{#1}}

\newcommand{\rllle}{\ar@{=}[lll]}
\newcommand{\rllleh}[1]{\ar@{=}[lll]^-{#1}}
\newcommand{\rlllem}[1]{\ar@{=}[lll]_-{#1}}
\newcommand{\rlllen}[1]{\ar@{=}[lll]|-{#1}}

\newcommand{\rlllue}{\ar@{=}[lllu]}
\newcommand{\rlllueh}[1]{\ar@{=}[lllu]^-{#1}}
\newcommand{\rllluem}[1]{\ar@{=}[lllu]_-{#1}}
\newcommand{\rllluen}[1]{\ar@{=}[lllu]|-{#1}}

\newcommand{\rllluue}{\ar@{=}[llluu]}
\newcommand{\rllluueh}[1]{\ar@{=}[llluu]^-{#1}}
\newcommand{\rllluuem}[1]{\ar@{=}[llluu]_-{#1}}
\newcommand{\rllluuen}[1]{\ar@{=}[llluu]|-{#1}}

\newcommand{\rllluuue}{\ar@{=}[llluuu]}
\newcommand{\rllluuueh}[1]{\ar@{=}[llluuu]^-{#1}}
\newcommand{\rllluuuem}[1]{\ar@{=}[llluuu]_-{#1}}
\newcommand{\rllluuuen}[1]{\ar@{=}[llluuu]|-{#1}}

\newcommand{\rlllde}{\ar@{=}[llld]}
\newcommand{\rllldeh}[1]{\ar@{=}[llld]^-{#1}}
\newcommand{\rllldem}[1]{\ar@{=}[llld]_-{#1}}
\newcommand{\rlllden}[1]{\ar@{=}[llld]|-{#1}}

\newcommand{\rllldde}{\ar@{=}[llldd]}
\newcommand{\rlllddeh}[1]{\ar@{=}[llldd]^-{#1}}
\newcommand{\rlllddem}[1]{\ar@{=}[llldd]_-{#1}}
\newcommand{\rllldden}[1]{\ar@{=}[llldd]|-{#1}}

\newcommand{\rlllddde}{\ar@{=}[lllddd]}
\newcommand{\rllldddeh}[1]{\ar@{=}[lllddd]^-{#1}}
\newcommand{\rllldddem}[1]{\ar@{=}[lllddd]_-{#1}}
\newcommand{\rlllddden}[1]{\ar@{=}[lllddd]|-{#1}}

\title{Composing Codensity Bisimulations}

\author[1, 2]{Mayuko Kori}
\affil[1]{National Institute of Informatics}
\affil[2]{The Graduate University for Advanced Studies, SOKENDAI}

\author[1, 2]{Kazuki Watanabe}

\author[3]{\\Jurriaan Rot}
\affil[3]{Radboud University}

\author[4]{Shin-ya Katsumata}
\affil[4]{Kyoto Sangyo University}

\begin{document}

\maketitle

\begin{abstract}
  Proving compositionality of behavioral equivalence on state-based systems with respect to algebraic operations is a classical and widely studied problem. We study a categorical formulation of this problem, where operations on state-based systems modeled as coalgebras can be elegantly captured through distributive laws between functors. To prove compositionality, it then suffices to show that this distributive law lifts from sets to relations, giving an explanation of how behavioral equivalence on smaller systems can be combined to obtain behavioral equivalence on the composed system.

In this paper, we refine this approach by focusing on so-called codensity lifting of functors, which gives a very generic presentation of various notions of (bi)similarity as well as quantitative notions such as behavioral metrics on probabilistic systems. The key idea is to use codensity liftings both at the level of algebras and coalgebras, using a new generalization of the codensity lifting. The problem of lifting distributive laws then reduces to the abstract problem of constructing distributive laws between codensity liftings, for which we propose a simplified sufficient condition. Our sufficient condition instantiates to concrete proof methods for compositionality of algebraic operations on various types of state-based systems.
We instantiate our results to prove compositionality of qualitative and quantitative properties of deterministic automata. We also explore the limits of our approach by including an example of probabilistic systems, where it is unclear whether the sufficient condition holds, and instead we use our setting to give a direct proof of compositionality.

In addition, we propose a compositional variant of  Komorida et al.'s codensity
games for bisimilarities.  A novel feature of this variant is that
it can also compose game invariants, which are subsets of winning
positions. Under our sufficient condition of the liftability of
distributive laws, composed games give an alternative proof of the
preservation of bisimilarities under the composition.
\end{abstract}

\section{Introduction}

Bisimilarity and its many variants are fundamental notions of behavioral equivalence on state-based systems. A classical question is whether a given notion of equivalence is \emph{compositional} w.r.t.\ algebraic operations on these systems, such as parallel composition of concurrent systems, product constructions on automata, or other language constructs. The problem can become particularly challenging when moving from bisimilarity to quantitative notions such as behavioral metrics on probabilistic systems, where it is already non-trivial to even define what the right notion of compositionality is (see, e.g.~\cite{DBLP:conf/concur/Lago023,DBLP:journals/corr/GeblerLT16}).

To formulate and address the problem of compositionality of algebraic operations on state-based systems at a high level of generality, a natural formalism is that of \emph{coalgebra}~\cite{DBLP:books/cu/J2016}, which is parametric in the type of system, as modeled by a ``behavior'' endofunctor $F$. A key idea, originating in the seminal work of Turi and Plotkin~\cite{DBLP:conf/lics/TuriP97}, is that composition operations and in fact whole languages specifying state-based systems arise as \emph{distributive laws} between $F$ and a functor (or monad) $T$ which models the syntax of a bigger programming language. In fact, Turi and Plotkin's \emph{abstract GSOS laws} are forms of such distributive laws which guarantee compositionality of strong bisimilarity, generalizing the analogous result for specifications in the \emph{GSOS format} for labeled transition systems~\cite{DBLP:journals/jacm/BloomIM95}. Distributive laws here precisely capture algebra-coalgebra interaction needed for compositionality.

The results of Turi and Plotkin specifically apply to strong bisimilarity, and not directly to other notions such as similarity or behavioral metrics. To prove compositionality in these cases, the key observation is that the task is to define the composition operation at hand not only on state-based systems but also on relations (or, e.g., pseudometrics). This idea has been formalized in the theory of coalgebras by requiring a \emph{lifting} of the distributive law that models the composition operator to a category of relations. The lifting of the behavior functor $F$ then specifies the notion of bisimilarity at hand, and the lifting of the syntax functor $T$ specifies the way that relations on components can be combined into bisimulations on the composite system. It has been shown in~\cite{DBLP:journals/acta/BonchiPPR17} at a high level of generality that the existence of such liftings, referred to in this paper as \emph{liftability}, ensures compositionality. The generality there is offered by the use of liftings of the behavior functor in fibrations which goes back to~\cite{DBLP:journals/iandc/HermidaJ98}, and which allows us to study not only relations but also, for instance, coinductively defined metrics or unary predicates (e.g.,~\cite{DBLP:journals/ngc/KomoridaKHKHEH22,DBLP:journals/logcom/SprungerKDH21,DBLP:journals/mscs/HasuoKC18,DBLP:conf/concur/Bonchi0P18}).
But the main challenge here is to \emph{prove liftability}; this is what we address in the current paper
(Question~\ref{q:codlift}.\ref{item:lift}).

We focus on the \emph{codensity lifting} of behavior functors~\cite{DBLP:journals/logcom/SprungerKDH21},
which allows to model a wide variety of coinductive relations and predicates including (bi)similarity but also, for instance,
behavioral metrics; these are commonly referred to as \emph{codensity bisimilarity}. In fact, the codensity lifting directly generalizes the celebrated Kantorovich distance between probability
distributions used to define, for instance, metrics between probabilistic systems~\cite{DBLP:journals/tcs/DesharnaisGJP04}. The codensity lifting is parametric
in a collection of modal operators, which makes it on the one hand very flexible and on the other hand much more structured than arbitrary liftings.
Moreover, the codensity lifting allows to characterize \emph{codensity games}~\cite{DBLP:journals/ngc/KomoridaKHKHEH22}, a generalization of classical bisimilarity games~\cite{DBLP:journals/igpl/Stirling99,DBLP:conf/qest/DesharnaisLT08}
from transition systems to coalgebras, and from strong bisimilarity to a wide variety of coinductively defined relations, metrics and even topologies.

In this paper we study \emph{compositionality of codensity bisimulations} with respect to composition operators on the underlying state-based systems,
modeled as coalgebras. We model these composition operators as distributive laws between an $n$-ary product functor and the behavior functor $F$, referred to as \emph{one-step composition operators}.
We mainly tackle two problems (Question~\ref{q:codlift}) in this paper:
1) How do we lift the $n$-ary product functor
in order to capture non-trivial combination of relations,
pseudometrics etc?
2) When is a one-step composition operator liftable?

A key idea for the first problem is to use codensity liftings not only to lift the behavior functor $F$ and get our desired notion of equivalence, but to use another codensity lifting of the product functor, to explain syntactically how relations (or pseudometrics, etc) should be combined from components to the composite system. This combination can be a simple product between relations, but in many cases, such as for behavioral metrics, it needs to be more sophisticated; the flexibility offered by the codensity lifting helps to define the appropriate constructions on relations.

The second problem about liftability
then becomes that of proving the existence of a \emph{distributive law between codensity liftings}. To this end, we exhibit a sufficient condition that ensure the existence of this distributive laws, defined in terms of two properties: 1) commutation between the underlying modalities that define codensity liftings, and 2) an \emph{approximation} property, which resembles the one used to prove expressiveness of modal logics in~\cite{DBLP:conf/lics/KomoridaKKRH21}. Underlying our approach is a combination of a new generalization of the codensity lifting beyond endofunctors, to allow to lift product functors, and the adjunction-based decomposition of codensity liftings proposed in~\cite{DBLP:conf/stacs/BeoharG0MFSW24}. The adjunction-based approach precisely allows us to arrive at our sufficient condition.

We instantiate our approach to pseudometrics on deterministic automata. We also revisit the compositionality of parallel composition w.r.t.\ behavioral metrics studied in~\cite{DBLP:journals/corr/GeblerLT16}. In this case it is not clear whether our sufficient condition holds; but the framework nevertheless helps to prove liftability and thereby compositionality.

We further study a compositional variant of codensity games.  A key
observation is that our \emph{composite codensity games} consist
of positions that are tuples of positions of codensity games for
component systems.  This design of the composite game directly leads to the
compositionality of game invariants, which characterize winning
positions. Assuming our sufficient condition for the liftability
of distributive laws, we present an alternative proof of the
preservation of bisimilarities along compositions (the inequality
in Corollary~\ref{cor:presbisim}).

In summary, the contributions of this paper are
\begin{itemize}
\item a generalization of the codensity lifting beyond endofunctors
  (Section~\ref{sec:gcodlift}), which can be used in a special case to
  lift products in various ways (Section~\ref{sec:codlift_structure});
\item a sufficient condition for proving the existence of a
  distributive law between codensity liftings
  (Section~\ref{sec:liftability}), with several detailed examples
  (Section~\ref{sec:examples});
\item a composition of codensity games, which also composes game
  invariants, and an alternative proof of the preservation of
  codensity bisimilarities under our sufficient condition
  (Section~\ref{sec:game});
\end{itemize}
After the preliminaries (Section~\ref{sec:preliminaries}), we start
the paper with a more detailed overview of our approach
(Section~\ref{sec:coalg_bisim}).

\paragraph{Related work} There is a wide range of results in the
process algebra literature on proving compositionality and on
rule formats that guarantee it,
usually focused on transition systems (see,
e.g.,~\cite{DBLP:journals/tcs/MousaviRG07} for an overview). A full
account is beyond the scope of this paper. We focus instead on
generality in the type of models and the type of coinductive
predicate. Concerning general coalgebraic frameworks for compositionality,
we have already mentioned Turi and Plotkin's
abstract GSOS format; the main innovation in the current paper is that
we go beyond bisimilarity by employing the codensity lifting.

In~\cite{DBLP:journals/acta/BonchiPPR17}, it is shown that liftings of
distributive laws to fibrations yield so-called \emph{compatibility},
a property that ensures soundness of up-to techniques, and which
implies compositionality. On the one hand, if one uses the so-called
canonical relation lifting then all distributive laws lift, but this
only concerns strong bisimilarity; on the other hand,
in~\cite{DBLP:journals/acta/BonchiPPR17} examples beyond bisimilarity
are studied but liftability there is proven on an ad-hoc basis. In the
current paper we identify the codensity lifting as a sweet spot
between the (restricted) canonical relation lifting and abstract,
unrestricted lifting of functors, and focus instead on conditions that
allow to prove liftability.

Many bisimilarity notions are known to be characterized by winning
positions of certain safety games, including bisimilarity on Kripke
frames~\cite{DBLP:journals/igpl/Stirling99}, probabilistic
bisimulation~\cite{DBLP:conf/icalp/FijalkowKP17,DBLP:conf/qest/DesharnaisLT08},
and bisimulation metric~\cite{DBLP:conf/concur/KonigM18}.  There are
several coalgebraic
frameworks~\cite{DBLP:conf/concur/KonigM18,DBLP:conf/lics/FordMSB022}
that captures such a relationship between games and bisimilarities.
In this paper, we focus on codensity
games~\cite{DBLP:journals/ngc/KomoridaKHKHEH22} that are naturally
obtained by codensity liftings and coalgebras.

Codensity liftings are first introduced to lift monads across
fibrations~\cite{DBLP:conf/calco/KatsumataS15}. They subsume the
Kantorivich metric on the Giry
monad. In~\cite{DBLP:conf/cmcs/SprungerKDH18}, they are extended to
lift endofunctors across $\clat$-fibrations over $\Set$, and are shown
to be a generalization of Baldan et al.'s Kantorovich
lifting~\cite{DBLP:conf/fsttcs/BaldanBKK14}. Prior to these papers,
the monadic property of the Kantorovich metric is studied in an
unpublished manuscript by van Breugel \cite{franck}.  Codensity
liftings for endofunctors can be done in other base categories (such
as $\mathbf{Meas}$ and $\mathbf{Vect}$) without modification; see
\cite{DBLP:journals/ngc/KomoridaKHKHEH22}.  In
\cite{DBLP:conf/fossacs/GoncharovHNSW23}, the Kantorovich liftings are
studied in the context of Lawvere metric space, and are related with
other notions of liftings, such as lax extensions and (monotone)
predicate liftings.

\section{Preliminaries}\label{sec:preliminaries}

For a mathematical entity $x$ equipped with the notion of product
$(\times)$, by $x^N$ we mean the $(\times)$-product of $N$-many copies
of $x$. For instance, when $\CC$ is a category, $\CC^N$ denotes the
product category of $N$-many copies of $\CC$.
By abuse of notation, we use the same letter $N$ for the set $\{1, \dots, N\}$
and write $i \in N$ for $i \in \{1, \dots, N\}$.

For two natural transformations $\alpha\colon F\to G,\beta\colon G\to H$, their
vertical composition is denoted by $\beta\bul\alpha$, which is the
componentwise composition of $\beta$ and $\alpha$.
For a functor $p\colon \EE\to\BB$ and an object $X\in \BB$, the {\em fiber
  category} over $X$, denoted by $\EE_X$, is the category whose objects are
$P\in\EE$ such that $pP=X$, and whose morphisms are $f$ such that
$pf=id_X$.

\subsection{$\clat$-fibration}\label{sec:clatfib}

A \emph{$\clat$-fibration} is a posetal fibration $p\colon \EE\to\BB$ such
that each fiber $\mathbb{E}_X$ is a complete lattice and each pullback
$f^*\colon \mathbb{E}_Y\rightarrow \mathbb{E}_X$ preserves all meets.
The order relation of a fiber $\mathbb{E}_X$ is denoted by
$\sqsubseteq$, the meet and the join by $\bigsqcap,\bigsqcup$, and the
empty meet and join by $\top, \bot$. We remark that $\clat$-fibrations
are a special case of {\em topological functors} \cite{herr74:topo}, where each
fiber is a small partial order.

Examples of $\clat$-fibrations are forgetful functors from the
following categories into $\Set$: 1) $\Pre$, the category of preorders
and monotone functions, 2) $\EqRel$, the category of equivalence
relations and relation-respecting functions, 3) $\Top$, the category
of topological spaces and continuous functions, 4) $\PMet$, the
category of pseudometric spaces and non-expansive functions,
5) $\ERelNA$, the category of sets with endorelations
and relation-respecting functions.
The forgetful functors are denoted by $\pfgt{\Pre}$ etc.

Every $\clat$-fibration $p\colon \EE\to\BB$ is faithful, i.e., for any
$\EE$-object $P,Q$ and $\BB$-morphism $f\colon pP\to pQ$, there is at most
one $\EE$-morphism $\dot f$ such that $p\dot f=f$.  When such an $\EE$-morphism
exists, we write $f\colon P\dto Q$. For instance, when $p$ is the
forgetful functor from $\Top$, the notation $f\colon (X,O_X)\dto (Y,O_Y)$
for a function $f\colon X\to Y$ is equivalent to the statement: ``$f$ is
continuous with respect to topologies $O_X,O_Y$''.  From the
fiberedness, $f\colon P\dto Q$ is equivalent to the inequality $P\sqsubseteq f^*Q$.

\subsection{Liftings}\label{sec:liftings}

We extensively use the concept of {\em lifting} of various categorical
structures in this paper.  Let $p\colon \EE\to\BB$ and $q\colon \FF\to\CC$ be
functors.  A {\em lifting} of a functor $F\colon \BB\to\CC$ along $p,q$ (or
simply $p$ when $p=q$) is a functor $\dF\colon \EE\to\FF$ such that
$q\circ \dF=F\circ p$. Similarly, for liftings $\dot F,\dot G$ of
$F,G\colon \BB\to\CC$ along $p,q$ and a natural transformation
$\alpha\colon F\Rightarrow G$, a {\em lifting} of $\alpha$ with respect to
$\dot F,\dot G$ is a natural transformation
$\dot\alpha\colon \dot F\Rightarrow\dot G$ such that
$q\circ \dot\alpha=\alpha\circ p$.

To extend the concept of lifting to other categorical structures, it
is convenient to set-up the {\em arrow 2-category $\CAT^\to$} given by
the following data: a 0-cell is a functor $p\colon \EE\to\BB$, a 1-cell from
$p\colon \EE\to\BB$ to $q\colon \FF\to\CC$ is a pair
$(F\colon \BB\to\CC,\dot F\colon \EE\to\FF)$ of functors such that $\dot F$ is a
lifting of $F$ along $p,q$, and a 2-cell from $(F,\dot F)\colon p\to q$ to
$(G,\dot G)$ is a pair of natural transformations
$(\alpha,\dot\alpha)$ such that $\dot\alpha$ is a lifting of $\alpha$
with respect to $\dot F,\dot G$. This 2-category $\CAT^\to$ has
2-products: the terminal 0-cell is $id_1$, and the binary product of
$p_1,p_2$ is the product functor $p_1\times p_2$ (we only use finite
ones in this paper). The evident forgetful 2-functor is denoted by
$\cod\colon \CAT^\to\to\CAT$. Any categorical structure expressible within
this 2-category corresponds to a pair of the categorical structure and
its lifting along functors (and vice versa).  For instance, a monad on
$p\colon \EE\to\BB$ in $\CAT^\to$ is a pair of a monad on $\BB$ together
with its lifting along $p$ (as a monad).

For the theory of coalgebraic bisimulation in
\sref{sec:coalgebraicbisim}, we often focus on the full-sub 2-category
$\CAT^{\clat}$ of $\CAT^\to$ obtained by restricting 0-cells to
$\clat$-fibrations.  We say that a 1-cell $(F,\dF)\colon p\to q$ in
$\CAT^\clat$ is {\em fibered} if $\dF$ preserves Cartesian morphisms,
or equivalently $(Ff)^*(\dF P) = \dF(f^*P)$ holds for any $f,P$.

The action of this 2-functor on hom-categories is
denoted by the functor
$\cod_{p,q}\colon \CAT^\to(p,q)\to\CAT(\cod(p),\cod(q))$. This is a
$\clat$-fibration, hence faithful. Thus, for a natural transformation
$\alpha\colon F\Rightarrow G$ and liftings $\dF, \dG$ of $F,G$ respectively,
there is at most one lifting $\dot\alpha\colon  \dF \Rightarrow \dG$ of
$\alpha$. When it exists, we say that $\alpha$ is {\em liftable} with
respect to $\dF,\dG$.
\begin{align*}
  &\CAT^\clat(p,q)((F,\dF),(G,\dG))\\
  & \cong
  \{\alpha\colon F\Rightarrow G~|~\fa{P\in\mathrm{dom}(p)}\alpha_{pP}\colon \dF P\dto\dG P\}.
\end{align*}
If $\alpha_{pP}\colon \dF P\dto\dG P$ holds for $P$ in a subclass $C$ of
objects in $\EE$, we say that $\alpha$ is liftable {\em on} $C$.

\section{Overview} \label{sec:coalg_bisim}

The theory of {\em coalgebras} provides a categorical framework for
expressing various transition systems in a unified manner.  Let $\CC$
be a category and $F\colon \CC\to\CC$ be an endofunctor (called {\em
  behavior functor}).  An {\em $F$-coalgebra} is a pair of an object
$X$ (called {\em carrier}) and a morphism $\coalg FcX$. A morphism
from $\coalg FcX$ to $\coalg FdY$ is a morphism $h\colon X\to Y$ such that
$Fh\circ c=d\circ h$.
\begin{mdexample}\label{ex:Kripke_coalg}
  We write $\pow$ for the covariant powerset functor on $\Set$.  A
  $\pow$-coalgebra $\coalg \pow cX$ bijectively corresponds to a
  binary relation on $X$, that is, a {\em Kripke frame}.
\end{mdexample}
\begin{mdexample}\label{ex:detaut_coalg}
  Fix an alphabet $\Sigma$. We define the endofunctor $\Fda$ on $\Set$
  by $\Fda\defeq 2\times(-)^\Sigma$. An $\Fda$-coalgebra
  $\coalg \Fda c X$ bijectively corresponds to a deterministic
  automata over alphabet $\Sigma$.
\end{mdexample}

The category of $F$-coalgebras and morphisms
between them is denoted by $\Coalg F$, and the evident forgetful
functor to $\CC$ is denoted by $\UCoalg F$. We note that
$(\Coalg F)^N\cong\Coalg{F^N}$. When $(F,\dF)\colon p\to p$ is a 1-cell in
$\CAT^\to$, the application of $p$ to $\dF$-coalgebras extends to a
functor $\Coalg p\colon \Coalg \dF\to\Coalg{F}$.

\subsection{Coalgebraic Bisimulation}

One of the important concepts in state transition systems is the
identification of two states that behave in the same way.
A classical notion of behavioral equivalence is {\em bisimilarity},
which has been
extensively studied in process algebra, coalgebra theory and modal logic.
Recently, the concept of coalgebraic bisimulation,
 initially formulated as a binary
relation, is expressed using other spatial structures, such as {\em
  pseudometrics} \cite{DBLP:journals/lmcs/BaldanBKK18} and {\em topologies} \cite{DBLP:conf/cmcs/SprungerKDH18}. A
uniform treatment of these spatial representations of bisimulation
have been developed using the framework of $\clat$-fibrations
\cite{DBLP:conf/cmcs/SprungerKDH18}. We employ the formulation of bisimulation as {\em
  coalgebras} in suitable fibrations. This formulation goes back to
the seminal work by Hermida and Jacobs \cite{DBLP:journals/iandc/HermidaJ98}, and the key
ingredient is the {\em lifting} of behavior functors along
fibrations.  We recall their theory here.

Let $p\colon \EE\to\BB$ be a $\clat$-fibration and $\coalg FcX$ be an
$F$-coalgebra. Hermida and Jacobs' theory first chooses a {\em lifting
  $\dF$} of $F$ along $p$; the main body of the definition of
bisimulation is packed into it. Then we formulate an {\em
  $\dF$-bisimulation} on a coalgebra $\coalg FcX$ to be a
$\dF$-coalgebra $\coalg{\dF}{\dot c}{P}$ such that $p(\dot c)=c$
(hence $p(P)=X$). The morphism $\dot c$ witnesses that $P$ is
respected by the underlying transition system $p(c)$.
An $\dF$-bisimulation
bijectively corresponds to $P\in\EE_X$ satisfying
$P\sqsubseteq c^*\circ\dF(P)$. We thus define the set of bisimulations on $c$ by
\begin{align*}
  \Bisim{\dF}{c}
  &\defeq
\{P\in\EE_X~|~P\sqsubseteq c^*\circ\dF(P)\},
\end{align*}
and impose a partial order on it by restricting the one on $\EE_X$.
Since $\EE_X$ is a complete lattice, it contains the {\em greatest}
postfixpoint $\nu(c^*\circ\dF)$ corresponding to {\em
  bisimilarity}.
\begin{mdexample}\label{ex:Kripke_bisim}
  We express the definition of bisimulation for $\pow$-coalgebras. We
  take the $\clat$-fibration $\pfgt{\EqRel}$ and take the following
  lifting $\dot \pow$ of $\pow$ along $\pfgt{\EqRel}$:
  \begin{align*}
    \dot \pow(X,P)\defeq
    (\pow(X),\{(U,V)~|~&\fa{x\in U}\ex{y\in V}(x,y)\in P\wedge\\
                       &\fa{x\in V}\ex{y\in U}(x,y)\in P\}).
  \end{align*}
  Then $P$ is a $\dot \pow$-bisimulation on $\coalg \pow cX$ if and
  only if it satisfies the standard bisimulation condition:
  \begin{align*}
    \fa{(x,y)\in P}&\fa{x'\in c(x)}\ex{y'\in c(y)}(x',y')\in P\wedge\\
                   &\fa{y'\in c(y)}\ex{x'\in c(x)}(x',y')\in P.
  \end{align*}
\end{mdexample}
\begin{mdexample}
  We represent language-equivalent states in deterministic automata by
  the $\clat$-fibration $\pfgt{\EqRel}$ and the following lifting
  $\dFda$ of $\Fda$ along $p$:
  \begin{displaymath}
    \dFda(P)\defeq
    \{((t_1, \rho_1), (t_2, \rho_2)) \mid
    t_1 = t_2\wedge
    \fa{a \in \Sigma}
    (\rho_1(a), \rho_2(a)) \in P \}.
  \end{displaymath}
  Let $P$ be an $\dFda$-bisimulation on $\coalg\Fda cX$. Then any pair
  of states $(x,y)\in P$ are language equivalent, that is, the set of
  words accepted from $x$ coincides with that of $y$.  Moreover,
  $\dFda$-bisimilarity coincides with the set of all
  language-equivalent state pairs.
\end{mdexample}

\subsection{Composing $F$-Coalgebras} \label{sec:comp_f_coalg}

Let $F\colon \CC\to\CC$ be an endofunctor over a category $\CC$. We model an
$N$-ary composition operation for $F$-coalgebras by a pair of
\begin{enumerate}
\item a functor $T\colon \CC^N\to\CC$ (called {\em \mergefunc functor})
  describing how the operation transforms coalgebra carriers, and
\item a natural transformation
  $\lambda\colon T\circ F^N \Rightarrow F\circ T$, describing the
  composition of a transition system by merging one-step transition of
  its arguments.
\end{enumerate}
The second component induces the lifting
$\funcdist T\lambda\colon (\Coalg F)^N\to \Coalg F$ of $T$ along
$(\UCoalg F)^N,\UCoalg F$ given by
\begin{displaymath}
  \funcdist T\lambda\defeq\lambda\circ T(c_1,\cdots,c_N).
\end{displaymath}
We call the pair $(T,\lambda)$ an {\em $N$-ary one-step composition
  operation} for $F$-coalgebras.
\begin{mdexample}\label{ex:Kripke_composition}
  We introduce a binary one-step composition operation
  $(\times,\lambda)$ on $\pow$-coalgebras. It takes the binary product
  of coalgebra carriers. The distributive law
  $\distpow_{X,Y}\colon \pow X\times\pow Y\to\pow(X\times Y)$ is given by
  $ \distpow_{X,Y}(A,B)\defeq A\times B  $.
\end{mdexample}
\begin{mdexample}\label{ex:detaut_composition}
  We introduce a binary one-step composition operation
  $(\times,\distFda)$ on $\Fda$-coalgebras. It takes the binary product
  of coalgebra carriers.  The distributive law
  $\distFda_{X,Y} \colon \Fda X\times \Fda Y\to \Fda (X\times Y)$ is
  given by
  \begin{displaymath}
    \distFda_{X, Y}((t_1, \rho_1), (t_2, \rho_2))
    \defeq (t_1 \land t_2,~ \lam{a}(\rho_1(a), \rho_2(a))).
  \end{displaymath}
\end{mdexample}

\subsection{Composing Coalgebraic Bismulations}
\label{sec:coalgebraicbisim}

The central theme of this paper is to deepen the understanding of the
interaction between composition operations on coalgebras and the
concept of bisimulation. The question we address is:
\begin{question}\label{q:theme}
  \begin{mdframed}
    Let $(T,\lambda)$ be an $N$-ary one-step composition operation for
    $F$-coalgebras.  How can we compose bisimulations
    $P_1,\cdots,P_N$ on $F$-coalgebras $c_1,\cdots,c_N$ into a
    bisimulation on the composed $F$-coalgebra
    $\coalgmerge T\lambda{c_1,\cdots,c_N}$?
  \end{mdframed}
\end{question}
In other words, the problem is about extending the composition
operation on $F$-coalgebras to bisimulations. Using Hermida and
Jacobs' coalgebraic formulation of bisimulation, we rephrase Question
\ref{q:theme} as follows:
\begin{question}
  \begin{mdframed}
    Let $(F,\dot F)\colon p\to p$ be a 1-cell in $\CAT^\clat$.  How can we
    lift $\funcdist T\lambda$ along
    $(\Coalg p)^N,\Coalg p$?
    \begin{equation}
      \vcenter{\xymatrix{
          (\Coalg\dF)^N \rdm{(\Coalg p)^N} \ar@{.>}[r]^-{?} & \Coalg\dF \rdh{\Coalg p} \\
          (\Coalg F)^N \rrh{\funcdist T\lambda} & \Coalg F
        }}\label{eq:coalglift}
    \end{equation}
  \end{mdframed}
\end{question}

Thanks to the coalgebraic formulation of bisimulation, we notice that
a lifting of the distributive law with respect to liftings of relevant
functors yields the operation \eqref{eq:coalglift} on
$\dF$-bisimulations that is compatible with the operation on
$F$-coalgebras. We formally state this principle that guides this
research as follows. This itself is not new, and a limited version is
seen as \cite[Proposition 6.3]{DBLP:journals/acta/BonchiPPR17}.
\begin{theorem}[composition of bisimulations]
  \label{th:liftop}
  Let $(F,\dF)\colon p\to p$ and $(T,\dT)\colon p^N\to p$ be 1-cells in $\CAT^\clat$
  and
  $(\lambda,\dot\lambda)\colon (T,\dT)\circ(F^N,\dF^N)\to
  (F,\dF)\circ(T,\dT)$ be a 2-cell. Then
  $(T_\lambda,\dot T_{\dot\lambda})\colon (\Coalg{p})^N\to \Coalg{p}$ is a
  1-cell in $\CAT^\to$.\qed
\end{theorem}
A restricted version of this theorem is also available.  Let $C$ be a
subclass of objects in $\EE$. If $\lambda$ is liftable on $C^N$, then
we obtain a 1-cell
$(T_\lambda,\dot T_{\dot\lambda})\colon (\Coalg p|_C)^N\to\Coalg p$; here
$\Coalg p|_C$ is the restriction of $\Coalg p$ to the category of
$\dF$-coalgebras over objects in $C$. This restriction will be used
after Proposition \ref{prop:local_comp_metric}.
\begin{corollary}[preservation of bisimilarities]\label{cor:presbisim}
  In the setting of Theorem \ref{th:liftop}, for any $F$-coalgebra
  $\coalg F{c_i}{X_i}$ ($i=1\cdots N$),
  \begin{displaymath}
    \pushQED{\qed}
    \dT(\nu(c_1^* \circ \dF),\cdots,\nu(c_N^* \circ \dF))\sqsubseteq \nu(T_\lambda(c_1,\cdots,c_N)^* \circ \dF).
  \qedhere
  \popQED
  \end{displaymath}
\end{corollary}
The preservation holds because $\nu(c^*\circ\dF)$ is the domain
  of the terminal object of $\Coalg{\dF}_c$ for each $F$-coalgebra
  $c$, and $\dot T_{\dot\lambda}$ maps coalgebras in
  $\Coalg{\dF}_{c_i}$ ($i \in N$) to a coalgebra in
  $\Coalg{\dF}_{T_\lambda(c_1,\cdots,c_N)}$.

This theorem merely says that when all ingredients
$F,\cdots,\dot\lambda$ are available, the lifting $\dT_{\dot\lambda}$
is available. In practice, the data $F,\dF,T,\lambda$ are given by
hand when designing transition systems and bisimulations on them,
while a problematic part is defining a lifting of the \mergefunc{}
$\dT$, and make the distributive law $\lambda$ {\em liftable} with
respect to $\dT\circ(\dF)^N$ and $\dF\circ\dT$.

Regarding lifting the behavior functor $F$, recently a systematic
method called the {\em codensity lifting} has been introduced
\cite{DBLP:journals/logcom/SprungerKDH21}. This is a generalization of
Kantorovich lifting \cite{DBLP:journals/lmcs/BaldanBKK18}, a classical notion of distance between
distributions, to general $\clat$-fibrations. The advantage of the
codensity lifting is its flexibility; it has some parameters and by
varying it we obtain various liftings. Yet, it provides more structure
than assuming arbitrary liftings of $F$ and $T$.
Another important point is that {\em codensity bisimulations}, where
$\dF$ is a codensity lifting, have an interesting game-theoretic
characterization. This flexibility and relevance to bisimulation games
is attractive, so we address the problem of having the codensity
liftings and the liftable $\lambda$.  This involves two specific
questions.  One is about extending the codensity lifting technique to
lift not only behavior functors, but also \mergefunc functors $T$.
\begin{question} \label{q:codlift}
  \begin{mdframed}
    Let $(T,\lambda)$ be an $N$-ary one-step composition operation for
    $F$-coalgebras.
    \begin{enumerate}
    \item \label{item:structure} How do we non-trivially lift
      \mergefunc{} functors (such as $N$-ary product functors) in order
      to capture composition operators at the level of relations?
    \item \label{item:lift} When is $\lambda$ liftable as a
      distributive law between codensity liftings of $F$ and $T$?
    \end{enumerate}
  \end{mdframed}
\end{question}
To answer to these questions, we employ Beohar et al.'s recent
decomposition result of codensity liftings
\cite{DBLP:conf/stacs/BeoharG0MFSW24}. Their decomposition
separates fibration-specific parts in the codensity lifting from a
central part that does the actual lifting task. Upon their
decomposition, in Section \ref{sec:codlift_structure}, we first address
(1) by generalizing the codensity lifting to arbitrary functors.
Then in Section \ref{sec:liftability}, we address (2) by providing a
sufficient condition to lift $\lambda$ to a distributive law
between codensity liftings. The proof of the sufficiency takes
  advantage of the decomposition and a 2-categorical nature of the
central part in the decomposition.

\section{Generalizing the Codensity Lifting}\label{sec:gcodlift}

We recall the multiple parameter codensity lifting in
\cite{DBLP:journals/logcom/SprungerKDH21}. It lifts an endofunctor
$F\colon \BB\to\BB$ along a $\clat$-fibration $p\colon \EE\to\BB$ using three
parameters $\O,\bO,\tau$. We pack $(p,\O,\bO)$ into the following
data.
\begin{definition}
  Let $A$ be a discrete category. A {\em \ptfib} is a tuple
  $(p\colon \EE\to\BB,\O\colon A\to\BB,\bO\colon A\to\EE)$ of functors such that
  $\O=p\circ\bO$ and $p$ is a $\clat$-fibration. The notation for
  a \ptfib{} is $\pfib p \EE \BB A \O \bO$.
\end{definition}
Let $\pfib p\EE\BB A\O\bO$ be a \ptfib. The codensity lifting of a functor $F\colon \BB\to\BB$ takes as
argument a natural transformation $\tau \colon  F \circ \O \to \O$. The components $\tau_a \colon F \circ \O(a) \to \O(a)$
are referred to in coalgebraic modal logic as \emph{modalities}; $\tau$ is just an indexed collection of these,
and below we will often refer to $\tau$ itself as a modality.
The {\em codensity lifting}
of $F\colon \BB\to\BB$ with $\tau$ along $p$ is given by
\begin{displaymath}
  \textstyle
  \codlift F\bO\tau X\defeq
  \bigsqcap_{a\in A,k\in\EE(X,\bO(a))}(\tau_a\circ F(pk))^*(\bO(a)).
\end{displaymath}
It is an endofunctor over $\EE$, and is a lifting of $F$ along $p$.
When $A=1$, it is called the {\em single-parameter} codensity lifting.
$\dF$-bisimulations (resp.~$\dF$-bisimilarities) where $\dF$ is a codensity lifting
are called {\em codensity bisimulations (resp.~codensity bisimilarities)}.
\begin{mdexample}\label{ex:codenlift_pow}
  We write $2$ for the two-point set $\{\fff,\ttt\}$ and $\Eq[2]$ for
  the object $(2,\{(x,x)~|~x\in 2\})$ in $\EqRel$.
We identify it as a functor of type $1\to\EqRel$, and form a
  \ptfib{} $\pfibU {\\EqRel}{\EqRel}\Set 1 2{\Eq[2]}$. The codensity lifting of
  $\pow$ with the modality $\diamond\colon \pow(2)\to 2$ given by
  $\diamond(U)=true\iff true\in U$ coincides with the lifting
  $\dot\pow$ in Example \ref{ex:Kripke_bisim}
  \cite{DBLP:journals/ngc/KomoridaKHKHEH22}.  Therefore the standard
  bisimulations on $\pow$-coalgebras are expressible as codensity
  bisimulations.
\end{mdexample}
\begin{mdexample}\label{ex:codlift_Fda}
  We adopt the following \ptfib \,
  $\pfibU p \EqRel \Set{\Sigma\uplus\{\epsilon\}} \O {\bO}$ where
	$\O$ and $\bO$ are functors constantly returning $2$ and $\Eq[2]$, respectively.
We define the
  modality for $\Fda$ by
  \begin{displaymath}
    \tau_{a}(t,\rho)=
    \begin{choice}
      t & (a=\epsilon) \\
      \rho(a) & (a\in\Sigma).
    \end{choice}
  \end{displaymath}
  Then the codensity lifting $\codlift{\Fda }{\bO}{\tau}$ maps
  $(X,P) \in \EqRel$ to the set $\Fda(X)=2 \times X^\Sigma$ paired with the following
  equivalence relation:
  \begin{displaymath}
    \{((t_1, \rho_1), (t_2, \rho_2)) \mid t_1 = t_2\wedge
    \fa{a \in \Sigma}
    (\rho_1(a), \rho_2(a)) \in P \}.
  \end{displaymath}
  The $\codlift{\Fda }{\bO}{\tau}$-bisimilarity on $c$ identifies the
  states that accepts the same language.
\end{mdexample}
To organize the discussion, we package all the ingredients into the
following data. It specifies both abstract transition systems and a
notion of bisimulation on them.
\begin{definition}
  {\em Codensity bisimulation data} consists of a \ptfib{}
  $\pfib p\EE\BB A\O\bO$, a functor $F \colon \BB \rightarrow \BB$, and a
  natural transformation $\tau\colon  F\circ\O\to\O$.
\end{definition}

\subsection{Decomposition of Codensity Lifting}

Let $(\pfib p\EE\BB A\O\bO,F,\tau)$ be codensity bisimulation data.
Recently, Beohar et al.~\cite{DBLP:conf/stacs/BeoharG0MFSW24}
introduced a decomposition of the codensity lifting as a ``sandwich''
of a monotone function within Galois connections. We adopt their
decomposition as it is useful for analyzing the interaction between
codensity liftings.  In this paper, we describe their decomposition in
fibered category theory.  The decomposition is given as
\begin{equation}
  \codlift F\bO\tau =
  \conc p\bO\circ
  \Sp[A]F\tau\circ
  \abs p\bO,\quad
  (\abs p\bO\dashv \conc p\bO)
  \label{eq:decomp}
\end{equation}
where $\Sp[A]F\tau$ is an endofunctor over a suitable category that we
introduce below. The corresponding equation in Beohar et al. is in
\cite[Section 4.4]{DBLP:conf/stacs/BeoharG0MFSW24}. Their presentation is based on indexed
lattices, and left and right adjoints are swapped.\footnote{ Our
  choice of left and right adjoint is consistent with the codensity
  lifting of {\em monads} (using Eilenberg-Moore algebra) \cite{DBLP:journals/lmcs/KatsumataSU18}.
}

The description of the components of \eqref{eq:decomp} is in order.
First, we observe that the product
$[A,\pfgt\PredNA^{op}] \colon [A,\PredNA^{op}] \to [A,\Set^{op}]$ of $A$-fold
copies of the opposite of the subobject fibration
$\pfgt\PredNA\colon \PredNA\to\Set$ is a $\clat$-fibration.  Now the
adjunction $\abs p\bO\dashv \conc p\bO$ in \eqref{eq:decomp} is given
between $\EE$ and the category obtained by the change-of-base of the
fibration $[A,\pfgt\PredNA^{op}]$ along the functor
$H(b)\defeq(\BB(b,\O(-)))^{op}$:
\begin{displaymath}
  \xymatrix@C=1cm@R=1.5em{
    \EE \adjunction{rr}{\abs p\bO}{\conc p\bO} \ar@/_.8pc/[rrd]_-p & & \Sp[A]\BB\O \pbcorner[ul] \rrrh{\ol{H}} \rdm{\Spf[A]\BB\O} & & [A,\PredNA^{op}] \rdh{[A,\pfgt\PredNA^{op}]} \\
    & & \BB \rrrm{H} & & [A,\Set^{op}] \\
  }
\end{displaymath}
Concretely, an object of $\Sp[A]\BB\O$ is a pair of an object
$X\in\BB$ and a mapping $S\colon A\to\Set$ such that for any $a\in A$,
$S(a)\subseteq \BB(X,\O(a))$. A morphism from $(X,S)$ to $(Y,T)$ is a
morphism $f\colon X\to Y$ such that for any $a\in A$ and $k\in T(a)$,
$k\circ f\in S(a)$.\footnote{This category is similar to the one of
  topological spaces. Put $A=1$. The mapping $S$ specifies a subset of
  $\BB(X,\O)$, which may be seen as a topology on $X$.  An element in
  $S$ corresponds to an open subset, and the condition on morphism
  corresponds to the backword preservation of open subsets. For
  this reason we name this category Sp.} The evident projection
functor $\Spf[A]\BB\O\colon \Sp[A]\BB\O\to\BB$ is a $\clat$-fibration,
because $s$ is so.

Beohar et al. showed that the following assignments form a left
adjoint functor~\cite[Theorem 7]{DBLP:conf/stacs/BeoharG0MFSW24}:
\begin{align*}
  \abs p\bO(P)&\defeq
                (p(P),\lam{a}p(\EE(P,\bO(a)))),&
                                                     \abs p\bO(f)\defeq f,
\end{align*}
and the object part of its right adjoint satisfies
\begin{displaymath}
  \textstyle
  \conc p\bO(X,S)
  =\bigsqcap_{a\in A,k\in S(a)}k^*(\bO(a)).
\end{displaymath}

The middle part $\Sp[A]\BB\O$ in \eqref{eq:decomp} is defined as a lifting
of an endofunctor $F\colon \BB\to\BB$ along $\Spf[A]\BB\O$ using a modality
$\tau\colon F\circ\O\to\O$.  It is an extension of the predicate lifting
$(P\colon X\to\Omega)\mapsto (\tau\circ F(P)\colon F(X)\to\Omega)$, which is
commonly used in the theory of coalgebras, to sets of morphisms into
$\Omega$.
\begin{align}
  \Sp[A]F\tau(X,S)&\defeq(FX,\lam{a}\{\tau_a\circ Fk~|~k\in S(a)\}) \label{eq:spfun1}\\
  \Sp[A]F\tau(f)&\defeq Ff. \label{eq:spfun2}
\end{align}
By composing the above concrete characterizations of $\abs p\bO(P)$,
$\conc p\bO(X,S)$ and $\Sp[A]F\tau$, it is immediate that the
codensity lifting satisfies the equation \eqref{eq:decomp}.

\subsection{Generalizing Codensity Liftings}\label{sec:gencod}

From Beohar et al.'s decomposition, we notice that the actual lifting
job of $F$ is done by $\Sp[A]F\tau$. We speculate that an extension of
$\mathrm{Sp}$ to natural transformations would also help lifting
distributive laws with respect to codensity liftings. We therefore
exhibit a 2-functorial nature of the $\mathrm{Sp}$ construction. Its
domain is the {\em lax coslice 2-category} $\laxsl A\CAT$ under a
discrete category $A$, which is defined as follows:
\begin{itemize}
\item A 0-cell is a pair of a category $\CC$ with a functor
  $\O\colon A\to\CC$.
\item A 1-cell from $(\CC,\O)$ to $(\DD,\P)$ is a pair of a functor
  $F\colon \CC\to\DD$ and a natural transformation $\tau\colon F\circ\O\to\P$.
\item A 2-cell from $(F,\tau)\colon (\CC,\O)\to(\DD,\P)$ to $(G,\upsilon)$
  is a natural transformation $\alpha\colon F\to G$ such that
\end{itemize}
Observe that an endo-1-cell on $(\CC,\O)$ bijectively corresponds to a
pair of an endofunctor $F$ on $\CC$ and a modality
$\tau\colon F\circ\O\to\O$.
\begin{theorem}\label{th:2-fun}
  The following assignments $\SP[A]$ form a 2-functor of type
  $\laxsl A\CAT\to\CAT^{\clat}$.
  \begin{align*}
    \SP[A](\CC,\O) & = \Spf[A]\CC\O
    & \SP[A](F,\tau) & = (F,\Sp[A]F\tau)
    & \SP[A](\alpha) & = \alpha
  \end{align*}
  Here, the functor $\Sp[A]F\tau$ is the extension of
  \eqref{eq:spfun1} to arbitrary 1-cell $(F,\tau)\colon (\CC,\O)\to(\DD,\P)$
  by the same formulas \eqref{eq:spfun1},\eqref{eq:spfun2}.  Moreover,
  $\Sp[A]F\tau$ preserves Cartesian morphisms and all meets.
  \qed
\end{theorem}
This construction suggests that a categorical structure in
$\laxsl A\CAT$ is transferred to $\clat$-fibrations equipped with the
same categorical structure and its lifting. We employ this property to
transfer distributive laws in $\laxsl A\CAT$ to $\CAT^\clat$ in
Theorem \ref{thm:suf_lift_dist}.

As the construction of $\Sp[A] F\tau$ is extended to arbitrary 1-cells
$(\CC,\O)\to(\DD,\P)$, it is natural to put {\em different} adjoints
in the decomposition \eqref{eq:decomp}. This leads us to the
following generalization:
\begin{definition} \label{def:gen_codlift} Let
  $\pfib p \EE \BB A \O \bO$ and $\pfib q \FF \CC A \P \bP$ be
  \ptfibs.  The {\em codensity lifting} of $F\colon \BB\to\CC$ along $p,q$ with a
  natural transformation $\tau\colon F\circ\O\to \P$, denoted by
  $\codlift F{\bO,\bP}\tau$, is defined by
  \begin{align}
    \codlift F{\bO,\bP}\tau
    &\defeq\conc q\bP \circ \Sp[A]F\tau \circ \abs p\bO \label{eq:decompgen} \\
    &\textstyle=\bigsqcap_{a\in A,k\in\EE(-,\bO(a))}(\tau_a\circ F(pk))^*(\bP(a)). \nonumber
  \end{align}
\end{definition}
We show some properties of this generalized codensity lifting.  The
following shows that it is the largest lifting that makes $\tau$
liftable with respect to $\dot F\circ\bO$ and $\bP$.  Here, two
liftings $\dF,\ddot F$ of $F$ are compared by:
$\dF\sqsubseteq \ddot F$ if $\dF(P)\sqsubseteq\ddot F(P)$ for any $P$.
\begin{theorem} \label{prop:codlift_decent}
  For any $(F,\dot F)\in\CAT^\clat(p,q)$,
  $\dot F\sqsubseteq \codlift F{\bO,\bP}\tau$ if and only if
  $\tau\colon \dot F\circ \bO\dto\bP$ in the fibration
  $\cod_{\Id_A,q}$ (\S{}\ref{sec:liftings}).
  \qed
\end{theorem}
This is the generalization of the universal property possessed by
codensity liftings \cite[Theorem
5.14]{DBLP:journals/logcom/SprungerKDH21}.

When the indexing set $A$ is one-point $1$, the codensity lifting
$\codlift{F}{\bO, \bP}{\tau}$ is determined by the object
$\codlift{F}{\bO, \bP}{\tau}\bO$.
We can utilize it for general cases $(A \neq 1)$
because a codensity lifting is a meet $\sqcap_{a \in A} \codlift{F}{\bO(a), \bP(a)}{\tau_a}$ of $A$,
and we have that
$\codlift{F}{\bO, \bP}{\tau}$ is determined by objects
$\codlift{F}{\bO(a), \bP(a)}{\tau_a}\bO(a)$.
\begin{proposition} \label{prop:codlift_omega} Assume $A=1$ in the
  setting of Definition~\ref{def:gen_codlift}.  Then
  $\codlift{F}{\bO, \bP}{\tau} = \conc{q}{\codlift{F}{\bO, \bP}{\tau}
    \bO} \circ \Sp[1]{F}{\mathrm{id}} \circ \abs{p}{\bO}$.
  \qed
\end{proposition}
\begin{proposition} \label{prop:f_codlift_omega} Assume $A=1$ in the
  setting of Definition~\ref{def:gen_codlift}.  Let $(F,\dF)\colon p\to q$ be a
  fibered one-cell. If $\conc{p}{\bO}\circ \abs{p}{\bO} = \mathrm{id}$ and
  $\dot{F}$ preserves fibered meets, the following statements hold.
\begin{enumerate}
  \item $\dot{F} = \conc{q}{\dot{F}\bO} \circ \Sp[1]{F}{\mathrm{id}} \circ \abs{p}{\bO}$.
  \item
  $\codlift{F}{\bO, \bP}{\tau} = \dot{F}$ if and only if $\codlift{F}{\bO, \bP}{\tau}\bO = \dot{F}\bO$.
  \qed
  \end{enumerate}
\end{proposition}
The assumption $\conc{p}{\bO}\circ \abs{p}{\bO} = \mathrm{id}$ is mild
enough for all examples of $\clat$-fibrations in this paper.  The
first statement means $\dot{F}$ has the same property as
$\codlift{F}{\bO, \bP}{\tau}$ written in
Proposition~\ref{prop:codlift_omega}, and it immediately yields the second
statement.  This proposition will be used to get a criterion for
analysing
$\codlift{F}{\bO, \bP}{\tau} = \times \colon \mathbb{E}^2 \to
\mathbb{E}$ (Corollary~\ref{cor:codlift_times}).

 \section{Codensity Liftings of \mergefunc Functors}
\label{sec:codlift_structure}

We use the generalized codensity lifting to lift the \mergefunc
functor.  Let $\pfib p \EE \BB A \O \bO$ be a \ptfib. Its $N$-fold
tupling $\pfib {p^N}{\EE^N}{\BB^N}A{\cptpl\O N}{\cptpl\bO N}$ is again
a \ptfib; here $\cptpl \O N$ denotes the tupling
$\angle{\O,\cdots,\O}\colon A\to\BB^N$ of $N$-fold copies of $\O$ (and the same for
$\cptpl\bO N$). We call the codensity lifting of a \mergefunc functor
$T\colon \BB^N\to\BB$ with $\sigma\colon T \circ \cptpl\O N\to\Omega$ along
$p^N,p$ the {\em $N$-codensity lifting}. Concretely,
\begin{align*}
  & \codlift T{\cptpl\bO N,\bO}\sigma(P_1,\cdots,P_N)\\
  & =\bigsqcap_{\substack{a\in A,\\ k_i\in\EE(P_i,\bO(a))}}
  (\sigma_a\circ T(pk_1,\cdots,pk_n))^*\bO(a).
\end{align*}
We overload the notation for the $N$-codensity lifting on that for the
codensity lifting: when $N$ can be read-off from $T$, we simply write
$\Ncodlift T{\bO}\sigma N$ to mean
$\codlift T{\cptpl\bO N,\bO}\sigma$.

\subsection{Product Functors by Codensity Liftings}

One of the most fundamental operators for composing processes and
transition systems is {\em parallel composition}. Typically it
generates a new transition system whose states are pairs $x||y$ of
component system's states.  This suggests that the carrier of the
composed system is the {\em binary product} of the carrier of its
components. We thus study the case where the \mergefunc functor $T$ is
$\times\colon \BB^2\to\BB$.

  As a side note,
  we can also handle other structure functors including coproducts, thanks to the generality of \S{}\ref{sec:gcodlift}.
Moreover, our framework facilitates the investigation of various modalities, even when we fix $T$ to the product functor, as illustrated in \S{}\ref{sec:pmet_codlift}.
This enables us to accommodate various non-trivial liftings for achieving compositionality, particulary in quantitative analysis (\S{}\ref{sec:examples}).

In this section we illustrate some liftings of the binary product
functor, and show that they can be expressed by the single-parameter
$2$-codensity lifting.  It is easy to apply results here to the
multiple-parameter codensity lifting because it is the intersection of
the single-parameter one.

\subsubsection{Product functor on the total category}

Suppose that the base category $\BB$ has a binary product functor
$\times\colon \BB^2\to\BB$. Since $p\colon \EE\to\BB$ is a $\clat$-fibration, the
functor $\dtimes\colon \EE^2\to\EE$ defined by
\begin{displaymath}
  P\dtimes Q \defeq \pi^*P \sqcap \pi'^*Q
  \quad
  (\text{$\pi,\pi'$ are 1st and 2nd projections})
\end{displaymath}
is a fibered lifting of $\times$ along $p^2,p$. It also preserves all
meets. In fact, it is the binary product functor on $\EE$
\cite[Proposition~9.2.1]{DBLP:books/daglib/0023251}.

When can $\dtimes$ be expressed by the codensity lifting
$\Ncodlift{\times}{\bO}{\sigma}{2}$?
Since Proposition~\ref{prop:f_codlift_omega} is applicable to $\dtimes$,
we obtain the following:
\begin{corollary} \label{cor:codlift_times} Let
  $\pfib p \EE \BB 1 \O \bO$ be a \ptfib{} such that $\BB$ has binary
  products, and
  $\conc{p}{\bO}\circ\abs{p}{\bO} = \mathrm{id}$.
Then for any modality $\sigma\colon \O\times\O\to\O$, we have
  $\dtimes=\Ncodlift{\times}{\bO}{\sigma}{2}$ if and only if
  $\Ncodlift{\times}{\bO}{\sigma}{2}(\bO, \bO) = \bO \dtimes \bO$.
  \qed
\end{corollary}
Therefore one can check $\dtimes=\Ncodlift{\times}{\bO}{\sigma}{2}$
only by checking the equality
$\Ncodlift{\times}{ \bO}{\sigma}{2}(\bO, \bO)=\bO\dtimes\bO$.

\begin{mdexample} \label{eg:ncod_times} Upon the same \ptfib{}
  $\pfibU p{\EqRel}\Set 1 2{\Eq[2]}$ in Example
  \ref{ex:codenlift_pow}, let $\land\colon 2 \times 2 \to 2$ be the
  logical conjunction. Then the binary product $\dtimes$ on $\EqRel$
  is the 2-codensity lifting $\Ncodlift{\times}{\Eq[2]}{\land}{2}$.
\end{mdexample}

\subsubsection{Metric lifting of the binary product by a
  modality} \label{sec:pmet_codlift}
Here we focus on the \ptfib \,
$\pfibU{p}{\mathbf{PMet}}{\Set}{1}{\intv}{(\intv,d_{\intv})}$ where
$\intv$ is the interval $[0, 1]$ and
$(\intv, d_{\intv}) \in \mathbf{PMet}$ is
the Euclidean distance over the interval on it.  We will use it for
quantitative bisimulations in \S{}\ref{sec:bisim_pseudo} and
\S{}\ref{sec:bisim_metric}.

We introduce a metric lifting of the binary product functor using a
binary function on the interval (which we also call a modality). We
show that it is equal to the codensity lifting given by the modality.
\begin{proposition} \label{prop:meet-free_pmet} Let
  $\sigma\colon \intv \times \intv \to \intv$ be a function. Given
  pseudometric spaces $(X,d_X),(Y,d_Y)$, we define the function
  $\lifttimes{\sigma}(d_X, d_Y)\colon (X \times Y)^2 \to \intv$ by
  \begin{displaymath}
    \lifttimes{\sigma}(d_X, d_Y)((x, y), (x', y')) \coloneqq \sigma(d_X(x, x'), d_Y(y, y')).
  \end{displaymath}
  Then
  this is a lifting of the product functor $\times\colon \BB^2 \to \BB$
along $\pfgt{\PMet}$
  if and only if $\sigma$ satisfies the following condition:
	\begin{equation} \label{eq:star}
		\left\{
		\begin{array}{l}
    \sigma\text{ is monotone}, \\
		\sigma(0, 0) = 0, \\
    \sigma(a, b) - \sigma(c, d) \leq \sigma(|a-c|, |b-d|).
		\end{array}
		\right.
	\end{equation}
Moreover,  if $\sigma$ satisfies \eqref{eq:star} then $\lifttimes{\sigma} = \Ncodlift{\times}{d_{\intv}}{\sigma}{2}$.
  \qed
\end{proposition}
For example, the modalities $\sigmant(a, b) \coloneqq 1-(1-a)(1-b)$, $\sigmap(a, b) \coloneqq (a+b)/2$, and $\sigma_\lor(a, b) \coloneqq \max(a, b)$ satisfy \eqref{eq:star},
\ifarxiv
see Appendix~\ref{ap:prop:sigma_negneg}.
\else
see \cite[Appendix A.2]{kori2024composing}.
\fi
The modalities $\sigmat(a, b) \coloneqq a \cdot b$ and $\sigma_\land(a, b) \coloneqq \min(a, b)$ do not satisfy the last condition of \eqref{eq:star}.
However, their codensity liftings are equal to those of $\sigmant$ and $\sigma_{\lor}$, respectively.
This follows from the following proposition.
\begin{proposition} \label{prop:isom_f}
  Let $f\colon \O \to \O$ be an isomorphism in $\BB$ satisfying $f^*\bO = \bO$.
  Then $\Ncodlift{T}{\bO}{\sigma}{N} = \Ncodlift{T}{\bO}{\sigma \circ T(f^N)}{N} = \Ncodlift{T}{\bO}{f \circ \sigma}{N}$.
  \qed
\end{proposition}
Putting $f(a) \coloneqq 1-a$, the proposition above implies
$\Ncodlift{\times}{d_{\intv}}{\sigmat}{2} =
\Ncodlift{\times}{d_{\intv}}{\sigmant}{2}$ and
$\Ncodlift{\times}{d_{\intv}}{\sigma_\land}{2} =
\Ncodlift{\times}{d_{\intv}}{\sigma_\lor}{2}$ because
$\sigmat = f \circ \sigmant \circ (f \times f)$ and
$\sigma_{\land} = f \circ \sigma_{\lor} \circ (f \times f)$ hold.

 \section{Codensity Lifting of One-Step
  Composition Operations} \label{sec:liftability}
  As discussed at the end of
\S{}\ref{sec:coalg_bisim}, our focus is on
\emph{modalities that lift one-step composition operations for coalgebras}.
Technically, this is equivalent to the liftability of distributive law
$\lambda\colon T \circ F^N \to F \circ T$ with respect to codensity
liftings of $F, T$.
\begin{definition}
  Let $(\pfib p\EE\BB A\O\bO,F,\tau)$ be codensity bisimulation
  data, and $(T,\lambda)$ be an $N$-ary one-step composition on
  $F$-coalgebras, see \S{}\ref{sec:comp_f_coalg} for the definition.
  We say that
  \emph{a modality $\sigma\colon T \circ \cptpl\O N \to\O$ lifts $\lambda$ along $p$}
if $\lambda$ is liftable
  with respect to $ \codlift T\bO\sigma\circ(\codlift F\bO\tau)^N$ and
  $\codlift F\bO\tau\circ\codlift T\bO\sigma$.
\end{definition}
Recall that the liftability of $\alpha$ may be restricted to a class
$C$ of objects in $\EE$. This restriction only appears in
\S\ref{sec:bisim_metric}.

\subsection{Sufficient Condition for Liftability}
We provide a sufficient condition for the modality
  $\sigma$ to lift a distributive law. Our approach employs the
lifting of the distributive law $\lambda$ by the 2-functor
$\mathrm{Sp}$ introduced in Theorem \ref{th:2-fun}, and the decomposed
definition of the generalized codensity lifting \eqref{eq:decompgen}.

First, we restrict $\sigma$ so that $\lambda$ becomes a 2-cell
$\laxsl A\CAT$ of type
\begin{displaymath}
  (T,\sigma)\circ(F^N,\tau^N)\Rightarrow (F,\tau)\circ(T, \sigma).
\end{displaymath}
Its image by the 2-functor $\mathrm{Sp}$ yields a lifting
$\dot\lambda$ of $\lambda$, whose type is
\begin{displaymath}
  \Sp[A]T\sigma \circ \Sp[A]{F^N}{\tau^N}\Rightarrow
  \Sp[A]F\tau \circ \Sp[A]T\sigma.
\end{displaymath}
This lifting enables
us to {\em interchange} the center of two codensity
liftings that meets in the middle of their composition (below we put
$R\defeq \conc p\bO, L\defeq\abs p\bO, R_N\defeq \conc {p^N}{\cptpl\bO
  N}, L_N\defeq\abs{p^N}{\cptpl\bO N}$):
\begin{align}
  &\Ncodlift{T}{\bO}{\sigma}{N}\circ(\codlift{F}{\bO}{\tau})^N \nonumber \\
  &= \Ncodlift{T}{\bO}{\sigma}{N}\circ\codlift{F^N}{\cptpl\bO N}{\tau^N}
    \hfill \quad ((\codlift{F}{\bO}{\tau})^N = \codlift{F^N}{\cptpl\bO N}{\tau^N}) \nonumber\\
  &=R \circ \Sp[A]T\sigma \circ L_N \circ R_N \circ \Sp[A]{F^N}{\tau^N} \circ L_N \nonumber \\
  &\Rightarrow R \circ \Sp[A]T\sigma \circ \Sp[A]{F^N}{\tau^N} \circ L_N
    \hfill \quad (\text{by }L_N \dashv R_N) \nonumber \\
  &\Rightarrow R \circ \Sp[A]F\tau \circ \Sp[A]T\sigma \circ L_N. \label{eq:lastline}
\end{align}
Overall, the above
natural transformation is again a lifting of $\lambda$, because the
last natural transformation is a lifting of $id$.

Next, we impose the last line \eqref{eq:lastline} to be equal to the
composition $\codlift{F}{\bO}{\tau}\circ\Ncodlift{T}{\bO}{\sigma}{N}$
of codensity liftings. Overall, we obtain a desired lifting of
$\lambda$ with respect to
$\Ncodlift{T}{\bO}{\sigma}{N}\circ(\codlift{F}{\bO}{\tau})^N$ and
$\codlift{F}{\bO}{\tau}\circ\Ncodlift{T}{\bO}{\sigma}{N}$.

\begin{theorem}[sufficient condition] \label{thm:suf_lift_dist}
  Let $(\pfib p\EE\BB A\O\bO,F,\tau)$ be codensity bisimulation data
  and $(T,\lambda)$ be an $N$-ary composition operation on
  $F$-coalgebras.  Then a modality $\sigma$ lifts $\lambda$ along $p$
  if the following conditions hold.
  \begin{enumerate}
  \item \label{item:2-cell} $\sigma$ makes $\lambda$ a 2-cell of type
    $(T,\sigma)\circ(F^N,\tau^N) \Rightarrow (F,\tau)\circ(T,\sigma)$
    in $\laxsl A\CAT$, and
  \item \label{item:approximating}
    $\conc p\bO \circ \Sp[A]F\tau \circ \Sp[A]T\sigma \circ
    \abs{p^N}{\cptpl\bO N} =
    \codlift{F}{\bO}{\tau}\circ\Ncodlift{T}{\bO}{\sigma}{N}$.
\qed
  \end{enumerate}
\end{theorem}
The first condition is equivalent to
$\sigma_a \circ T(\fami{\tau_a}{i}{N}) = \tau_a \circ F\sigma_a \circ \lambda$ for each $a \in A$,
and it induces $\Sp[A]T\sigma \circ \Sp[A]{F^N}{\tau^N} = \lambda^* \Sp[A]F\tau \circ \Sp[A]T\sigma$.

We investigate the second condition.  First, it can be expressed as
``$\Sp[A]T\sigma \circ \abs{p^N}{\cptpl\bO N}(P)$ is approximating to
$\codlift{F}{\bO}{\tau}$ for each $P \in \EE^N$'' by the concept of
{\em approximating families} introduced by Komorida et
al.~\cite{DBLP:conf/lics/KomoridaKKRH21}
for expressivity of coalgebraic modal logics
(in a restricted setting where
$\bO$ is constant on $A$).

\begin{definition} \label{def:approximating}
An object $S \in \Sp[A]{\BB}{\O}$ is \emph{approximating} to
  the codensity lifting $\codlift{F}{\bO}{\tau}$ if
  \begin{displaymath}
    \conc p\bO \circ \Sp[A]F\tau(S) \sqsubseteq \conc p\bO \circ \Sp[A]F\tau\circ \abs{p}{\bO} \circ \conc{p}{\bO}(S).
  \end{displaymath}
That is,
  for each $a \in A$ and $k\in \EE(\bigsqcap_{a' \in A, k' \in S_{a'}} k'^*\bO(a'), \bO(a))$,
$\bigsqcap_{a' \in A, k' \in S_{a'}} (\tau_{a'} \circ Fk')^*\bO(a') \sqsubseteq (\tau_a \circ Fpk)^*\bO(a)$.
\end{definition}
It is equivalent to
$\conc p\bO \circ \Sp[A]F\tau(S) = \conc p\bO \circ \Sp[A]F\tau\circ \abs{p}{\bO} \circ \conc{p}{\bO}(S)$ by the counit of $\abs{p}{\bO} \dashv \conc {p}{\bO}$.
Thus \eqref{item:approximating} in Theorem~\ref{thm:suf_lift_dist} can be written as
``$\Sp[A]T\sigma \circ \abs{p^N}{\cptpl\bO N}(P)$ is approximating to $\codlift{F}{\bO}{\tau}$ for each $P \in \EE^N$''.

Next we show
a result
for the approximating condition
that is applicable
when we can express
$k \in \EE(\bigsqcap_{a' \in A, k' \in S_{a'}}k'^*\bO(a'), \bO(a))$
by a subset of $S_a$.

\begin{proposition} \label{prop:lift_lambda_join}
  Assume that
each hom-poset $(\BB(X,\Omega),\leq_X)$ is a complete lattice,
  and write $\bigvee$ for the join.
  Then $(X, S) \in \Sp[A]{\BB}{\O}$ is approximating to $\codlift{F}{\bO}{\tau}$
  if all of the following hold, for each $a \in A$:
  \begin{enumerate}
  \item For each $S' \subseteq \mathbb{C}(X', \Omega(a))$,
    $\tau_a \circ F \bigvee_{f \in S'}f = \bigvee_{f \in S'} \tau_a
    \circ Ff$.
  \item For each $S' \subseteq \mathbb{C}(X', \Omega(a))$,
    $\bigsqcap_{f \in S'} f^*\bO \sqsubseteq (\bigvee_{f \in S'} f)^*
    \bO(a)$.
  \item For each
    $k\in \EE(\bigsqcap_{a' \in A, k' \in S_{a'}}k'^*\bO(a'),
    \bO(a))$, there exists $S'_k \subseteq S_a$
    s.t.~$pk=\bigvee_{k' \in S'_k} k'$.
  \end{enumerate}
  Additionally, if these conditions hold,
  $(X, S_a) \in \Sp[1]{\BB}{\O(a)}$ is approximating to $\codlift{F}{\bO(a)}{\tau_a}$
  for each $a \in A$.
  \qed
\end{proposition}

As a side note, when we prove the sufficient condition by
Proposition~\ref{prop:lift_lambda_join}, we get the liftability of
$\lambda$ with respect to the composition of not only $\Ncodlift{T}{\bO}{\sigma}{N}$ and
$\codlift{F}{\bO}{\tau}$ but also $\Ncodlift{T}{\bO}{\sigma}{N}$ and
$\codlift{F}{\bO(a)}{\tau_a}$ for each $a \in A$.  The latter is
stronger than the former.  It enables us to compose bisimulations with
respect to each $a \in A$.
\begin{lemma}
  Let $\alpha$ be a natural transformation $F \Rightarrow G$.
  Then $\alpha$ is a 2-cell $(F,\tau) \Rightarrow (G,\upsilon)$ in $\laxsl A\CAT$
  if and only if
  $\alpha$ is a 2-cell $(F, \tau_a) \Rightarrow (G, \upsilon_a)$ in $\laxsl 1 \CAT$ for each $a \in A$.
  \qed
\end{lemma}
\begin{proposition} \label{prop:lift_dist_a} If a distributive law
  $\lambda\colon T \circ F^N \Rightarrow F \circ T$ is liftable
  w.r.t.~$\Ncodlift{T}{\bO}{\sigma}{N}\circ (\codlift{F}{\bO(a)}{\tau_a})^N$
  and
  $\codlift{F}{\bO(a)}{\tau_a} \circ \Ncodlift{T}{\bO}{\sigma}{N}$ for each
  $a \in A$, then $\sigma$ lifts $\lambda$.
\qed
\end{proposition}
\subsubsection{Example: Bisimulation for $\pow$-coalgebras}
\label{eg:comp_kripke}
This example aims to show the composition of standard bisimulations.

  \paragraph{Codensity bisimulation data} We
  take the \ptfib{} $\pfibU p\EqRel\Set 1{2}{\Eq[2]}$ and the covariant
  powerset functor $\pow$ (Example \ref{ex:Kripke_coalg}) with the
  modality $\diamond\colon \pow 2 \to 2$ (Example
  \ref{ex:codenlift_pow}) to form codensity bisimulation data.  We
  have $\codlift\pow {\Eq[2]}\diamond=\dot\pow$ (Example
  \ref{ex:Kripke_bisim}).

  \paragraph{Binary one-step composition operation for
    $\pow$-coalgebras} We adopt $(\times,\distpow)$ in Example
  \ref{ex:Kripke_composition}; recall that the distributive law
  $\distpow$ is given by $ \distpow_{X,Y}(A, B)=A\times B$.

  \paragraph{Modality for $\times$} We adopt
  the modality $\land\colon 2\times2\to 2$ in Example \ref{eg:ncod_times}. The codensity lifting
  $\codlift\times {\Eq[2]} \land$ coincides with the binary product of
  predicates as in Example~\ref{eg:ncod_times}.

  \begin{proposition}\label{pp:kripke_prod_erel}
    The modality $\wedge$ lifts $\distpow$ along $\pfgt{\EqRel}$.
    \qed
  \end{proposition}
  This proposition, together with Corollary~\ref{cor:presbisim}, implies
  that for two $\pow$-coalgebras $\coalg\pow {c_1}{X_1}$ and
  $\coalg\pow{c_2}{X_2}$, if $x_i$ and $x_i'$ are
  $\codlift\pow{\Eq[2]}\diamond$-bisimilar in $c_i$ ($i \in \{1, 2\}$),
  the pairs $(x_1, x_2)$ and $(x_1', x_2')$ are
  $\codlift\pow{\Eq[2]}\diamond$-bisimilar in the composed $\pow$-coalgebra
  $\distpow \circ (c_1 \times c_2)$.

  The proof of Proposition \ref{pp:kripke_prod_erel} is the following:
  $\wedge$ satisfies two conditions in Theorem~\ref{thm:suf_lift_dist}:
  1) For each $A, B \subseteq 2$,
  $\diamond A \land \diamond B = \diamond \{a \land b \mid a \in A, b
  \in B\}$.  2) As discussed before Definition~\ref{def:approximating}, it is
  equivalent to approximating objects
  $\{\land \circ (k_P \times k_Q) \mid k_P \colon P \to \Eq[2], k_Q
  \colon Q \to \Eq[2]\}$ for each $P, Q \in \mathbf{EqRel}$.  We can
  prove it by Proposition~\ref{prop:lift_lambda_join} with the ordered
  object $(2, \leq)$ where $\leq$ is the order satisfying
  $\mathrm{false} \leq \mathrm{true}$.  One easily verifies the first
  and second condition in Proposition~\ref{prop:lift_lambda_join}.  For the
  third condition, for each $a \in A$ and
  $k\colon \Ncodlift{\times}{\Eq[2]}{\land}{N}(P, Q) \to \Eq[2]$, we
  define
  $S_k' \coloneqq \{\land \circ (\chi_{[x]_P} \times \chi_{[y]_Q})
  \mid k(x, y) = 1 \}$ where $\chi_{[x]_P}, \chi_{[y]_Q}$ are
  characteristic functions.  Then $pk = \bigvee_{k' \in S'_k} k'$
  holds.

\subsubsection{Example: Language Equivalence of Deterministic Automata}
\label{eg:lang_equiv}
We next capture language equivalence of deterministic automata using
  codensity bisimilarity.
  \paragraph{Codensity bisimulation data}
  They are set-up in Example \ref{ex:codlift_Fda}. We package the
  \ptfib{} $\pfibU p \EqRel \Set{\Sigma\uplus\{\epsilon\}} \O {\bO}$,
  the behavior functor $\Fda$ and the modality $\tau$ there into
  codensity bisimulation data.

  \paragraph{Binary one-step composition operation for $\Fda$-coalgebras}
  We adopt the one $(\times,\distFda)$ in Example \ref{ex:detaut_composition}.

  \paragraph{Modality for $\times$}
  We adopt the modality $\sigma$ given by $\sigma_a\defeq \land$ (the
  logical conjunction; here $a\in \Sigma\uplus\{\epsilon\}$).  The
  binary codensity lifting $\Ncodlift{\times}{\Eq[2]}{\land}{2}$ is equal
  to that in Example~\ref{eg:ncod_times}, so that
  $\Ncodlift{\times}{\Eq[2]}{\land}{2} = \dtimes$.
  \begin{proposition}\label{pp:langeq}
    The modality $\sigma$ lifts $\distFda$ along $\pfgt{\EqRel}$.
    \qed
  \end{proposition}
  As a consequence, given deterministic automata
  $\coalg \Fda {c_1}{X_1}$ and $\coalg \Fda {c_2}{X_2}$, if
  $x_i, x_i' \in X_i$ are language equivalent in $c_i$ ($i\in\{1,2\}$),
  the pairs $(x_1, x_2)$ and $(x_1', x_2')$ are also language equivalent
  in the product automaton $\distFda \circ (c_1 \times c_2)$.

  The proof of Proposition \ref{pp:langeq} is as follows. The modality
  $\land$ satisfies two conditions in Theorem~\ref{thm:suf_lift_dist}: 1)
  For each $(t, \rho), (t', \rho') \in 2 \times 2^\Sigma$, both
  $\land \circ \tau_a$ and
  $\tau_\epsilon \circ (2 \times \land^\Sigma) \circ \distFda$ map
  $((t, \rho), (t', \rho'))$ to $t \land t$ if $a = \epsilon$ and
  $\rho(a) \land \rho'(a)$ otherwise.  2) We can prove it by
  Proposition~\ref{prop:lift_lambda_join}.  The second and third conditions
  are the same as \S{}\ref{eg:comp_kripke}.  The first condition in Proposition~\ref{prop:lift_lambda_join}
  holds because for each $S' \subseteq \mathbf{Set}(X', 2)$,
both
  $\tau_a \circ 2 \times (\bigvee_{f \in S'}f)^\Sigma$ and
  $\bigvee_{f \in S'}\tau_a \circ 2 \times f^\Sigma$ map $(t, \rho) \in 2 \times X'^\Sigma$ to
  $t$ if $a=\epsilon$, and $\bigvee_{f \in S'} f(\rho(a))$ otherwise.

\subsection{Transferring Liftable Modalities} \label{sec:transfer}

We next study the \emph{transfer} of liftable modalities from
one $\clat$-fibration to another.
We first see the interaction between $L/R$ and the
$\mathrm{Sp}$.
\begin{lemma} \label{lem:full_fibred}
  Let $(F, \dF) \in \CAT^{\clat} (p, q)$.
  \begin{enumerate}
  \item $\abs q{\dF\bO}\circ\dF \sqsubseteq \Sp[A]F{id}\circ\abs p\bO$.
    They are equal when $\dF$ is full.
  \item $\dF \circ \conc p\bO \sqsubseteq \conc q{\dF\bO} \circ \Sp[A] F{id}$.
    They are equal when $\dF$ is fibered and preserving fibered meets.
    \qed
  \end{enumerate}
\end{lemma}
When $F=Id$, these properties are working behind Komorida et al.'s
argument of transferring codensity
bisimilarities~\cite{DBLP:journals/ngc/KomoridaKHKHEH22}.  Let
$(Id,G)\colon p\to q$ be a 1-cell in $\CAT^\clat$ such that $G$ is full,
fibered and preserving fiberwise meets; such a 1-cell is called a
\emph{transfer situation}
in~\cite{DBLP:journals/ngc/KomoridaKHKHEH22}.  Then the 1-cell
commutes with codensity liftings, that is,
$G \circ \codlift{F}{\mathbf{\Omega}}{\tau} = \codlift{F}{G \circ
  \mathbf{\Omega}}{\tau}\circ G$. From this commutativity, together
with the adjoint lifting theorem of Hermida and Jacobs \cite[Theorem
2.14]{DBLP:journals/iandc/HermidaJ98} (see also
\cite{DBLP:journals/logcom/SprungerKDH21,DBLP:journals/ngc/KomoridaKHKHEH22,DBLP:conf/calco/TurkenburgBKR23}
for relevant results), the preservation of codensity bisimilarity
is obtained:
\begin{displaymath}
  G(\nu (c^* \circ \codlift{F}{\mathbf{\Omega}}{\tau}))
  =\nu (c^*\circ\codlift{F}{G \circ \mathbf{\Omega}}{\tau})
  \quad
  (\coalg FcX).
\end{displaymath}
Such $G$ not only transfers the codensity bisimilarities, but also
the liftability property of modalities.
\begin{proposition}
  \label{prop:transfer_principle_decent}
  Let $(\pfib p\EE\BB A\O\bO, F, \tau)$ be codensity bisimulation data
  and $(T,\lambda)$ be an $N$-ary composition operation for
  $F$-coalgebras.  Let $(Id,G)\colon p\to (q\colon \FF\to\BB)$ be a 1-cell in
  $\CAT^\clat$ such that $G$ is full, fibered and preserving fibered
  meets.  If a modality $\sigma$ lifts $\lambda$ along $p$, then
  $\sigma$ lifts $\lambda$ also along $q$
  in the \ptfib{} $(\pfib q\FF\BB A\O{G\circ \bO})$.
\end{proposition}
The proof uses the equality
$G\circ\codlift T\bO\sigma=\codlift T{\cptpl\bO N}\sigma\circ G^N$
derivable from Lemma \ref{lem:full_fibred}.
\begin{example}
  Consider the $\clat$-fibration
  $\pfgt{\mathbf{ERel}}\colon \mathbf{ERel} \to \Set$ and the
  inclusion functor
  $i\colon \mathbf{EqRel} \hookrightarrow \mathbf{ERel}$.  Then
  $(\pfgt{\EqRel}, \pfgt{\ERelNA}, i)$ is a transfer situation.
In \S{}\ref{eg:comp_kripke}, we showed that $\sigma_{\land}$
lifts $\lambda^\pow$ along $\pfgt{\EqRel}$.
  Proposition~\ref{prop:transfer_principle_decent} induces that
  $\sigma_{\land}$
lifts $\lambda^\pow$ also along $\pfgt{\ERelNA}$.
\end{example}

 \section{Examples of Lifting Distributive laws via Modalities} \label{sec:examples}

\subsection{Bisimilarity Pseudometric for Deterministic
  Automata} \label{sec:bisim_pseudo} We next study a {\em bisimilarity
  pseudometric} for deterministic
automata~\cite{DBLP:journals/lmcs/BaldanBKK18,DBLP:conf/concur/Bonchi0P18},
which is a quantitative extension of language equivalence.
We fix a weight $w \in \intv$ where $\mathbb{I}$ is the interval $[0, 1]$.
We regard the interval $\intv$ as a
quantitative extension of binary truth values via the cast function
$i\colon 2\to\intv$ defined by $i(\ttt)=0$ and $i(\fff)=1$.

\paragraph{Codensity bisimulation data}
We use the same coalgebraic formulation of deterministic automata as
$\Fda $-coalgebras as \S{}\ref{eg:lang_equiv}
where the functor $\Fda$ is introduced in Example~\ref{ex:codlift_Fda}.
 The difference from
that example is that we employ the category of pseudometric spaces for
modeling bisimilarity pseudometrics. We adopt the \ptfib{}
$\pfibU p \PMet \Set {\Sigma\uplus\{\epsilon\}} \O \bO$ where
$\bO$ constantly returns the Euclidean space $(\intv,d_{\intv})$
over the $[0,1]$-interval. We pair it with the behavior
functor $\Fda$ for deterministic automata and the following modality
$\tau_a\colon \Fda(\intv) \to \intv$ to obtain codensity
bisimulation data:
\begin{displaymath}
  \tau_a(t,\rho)=
  \begin{choice}
    t & a=\epsilon \\
    w \cdot \rho(a) & a\in\Sigma.
  \end{choice}
\end{displaymath}

The codensity lifting $\codlift{\Fda }{\bO}{\tau}$ maps
$(X,d)\in\PMet$ to the space over $\Fda X$ with the following
pseudometric:
\begin{displaymath}
  ((t_1, \rho_1), (t_2, \rho_2)) \mapsto
  max\{|t_1 - t_2|, w \cdot \max_{a \in \Sigma} d(\rho_1(a), \rho_2(a))\}.
\end{displaymath}
Now, $\codlift{\Fda }{\bO}{\tau}$-bisimilarity on an $\Fda $-coalgebra
$\coalg \Fda cX$ is a quantitative extension of language equivalence
on $X$. It maps $(x, x') \in X^2$ to $0$ if the languages of $x, x'$
are the same, and $w^n$ otherwise where $n$ is the minimum length of a
word that is accepted from one and not from the other.  This notion
corresponds to language equivalence if $w=1$.

\paragraph{Binary one-step composition operation}
We reuse $(\times,\distFda)$
in Example \ref{ex:detaut_composition}.

\paragraph{Modality for $\times$}
We adopt $\sigma_\land\colon \intv^2 \to \intv$ given in
\S{}\ref{sec:pmet_codlift}.  As we saw in \S{}\ref{sec:pmet_codlift},
the binary codensity lifting
$\Ncodlift{\times}{d_{\intv}}{\sigma_\land}{2}$ maps
$(X,d_1),(X,d_2)\in \mathbf{PMet}$ to the
space on $X \times Y$ with the following pseudometric:
\begin{displaymath}
  ((x, y), (x', y')) \mapsto \max(d_1(x, x'), d_2(y, y')).
\end{displaymath}
\begin{proposition}
  The modality $\sigma_\land$ lifts $\distFda$ along $\pfgt{\PMet}$.
  \qed
\end{proposition}
This liftability yields a bound of bisimilarity pseudometrics of
composite automata: given deterministic automata
$c_1\colon X_1 \to \Fda (X_1)$ and $c_2\colon X_2 \to \Fda (X_2)$, the
distance between $(x, y)$ and $(x', y')$ in the product automaton
$\distFda \circ (c_1 \times c_2)$ is bounded by $\max(v_1, v_2)$ where
$v_1$ is the distance between $x, x'$ in $c_1$ and $v_2$ is the
distance between $y, y'$ in $c_2$.

The proof is to show that two conditions in Theorem~\ref{thm:suf_lift_dist}. See
\ifarxiv
Appendix~\ref{ap:bisim_pseudo}
\else
\cite[Appendix A.3.1]{kori2024composing}
\fi
for its proof.

\subsection{Similarity Pseudometric for Deterministic Automata} \label{sec:sim_pseudo}
Let us consider similarity on deterministic automata, as a variation of the example in the previous subsection.

We get a similarity framework just by relaxing the symmetry condition of $\PMet$
and taking an asymmetric truth-value domain instead of the Euclidean distance $d_{\intv}$.
For omitted proofs, see
\ifarxiv
Appendix~\ref{ap:sim_pseudo}.
\else
\cite[Appendix A.3.2]{kori2024composing}.
\fi

\paragraph{Codensity bisimulation data}
Let $\QPMet$ be the category of Lawvere metric spaces (meaning
asymmetric pseudometric spaces) and non-expansive maps.  The forgetful
functor $\pfgt{\QPMet}\colon \QPMet \to \Set$ is a $\clat$-fibration like
$\pfgt{\PMet}$.  We adopt the \ptfib{}
$\pfibU p \QPMet \Set A \O \bO$ such that $\bO$ constantly returns the
space over the interval with the Lawvere metric
$\das(x, y) \coloneqq \max(0, y-x)$.  We pair it with the behavior
functor $\Fda$ for deterministic automata and the same modality $\tau$
in \S{}\ref{sec:bisim_pseudo} to form codensity bisimulation data.

The codensity lifting $\codlift{\Fda }{\bO}{\tau}$ maps
$(X,d)\in\QPMet$ to the space on $\Fda(X)$ with the following
Lawvere metric:
\begin{displaymath}
  ((t_1, \rho_1), (t_2, \rho_2)) \mapsto
  max\{\das(i(t_1), i(t_2)), w \cdot \max_{a \in \Sigma} d(\rho_1(a), \rho_2(a))\};
\end{displaymath}
recall that $i\colon 2\to\intv$ is the cast function
(\S{}\ref{sec:bisim_pseudo}).  Note that $\das$ maps $(\ttt, \fff)$ to
$1$ and the others to $0$.  In \S{}\ref{sec:bisim_pseudo} we used the Euclidean distance
$d_{\intv}$ instead of $\das$, which also maps $(\fff, \ttt)$ to
$1$.

The choice of the Euclidean distance makes a difference in the notion of bisimilarity.
Now $\codlift{\Fda }{\bO}{\tau}$-bisimilarity on an $\Fda $-coalgebra
$\coalg \Fda cX$ is a similarity pseudometric on $X$ for the deterministic automaton $c$.
It maps $(x, x') \in X^2$ to $0$ if the language of $x$ is included in that of $x'$, and $w^n$ otherwise where $n$ is the minimum length of a
word that is accepted from $x$ and not from $x'$.
A bisimialrity pseudometric in \S{}\ref{sec:bisim_pseudo} detects words that are accepted from $x'$ and not from $x$ while a similarity pseudometric doesn't.

\paragraph{Binary one-step composition operation}
We reuse $(\times,\distFda)$
in Example \ref{ex:detaut_composition}.

\paragraph{Modality for $\times$}
We adopt $\sigma_\land\colon \intv^2 \to \intv$ given in
\S{}\ref{sec:pmet_codlift}.
The binary codensity lifting
$\Ncodlift{\times}{\das}{\sigma_\land}{2}$ maps $(X,d_1)$ and
$(Y, d_2)$ to the Lawvere metric on $X \times Y$
\begin{displaymath}
  ((x, y), (x', y')) \mapsto \max(d_1(x, x'), d_2(y, y')).
\end{displaymath}
\begin{proposition}
  The modality $\sigma_\land$ lifts $\distFda$ along $\pfgt{\QPMet}$.
  \qed
\end{proposition}
This liftability yields a bound of similarity pseudometrics of
composite automata: given deterministic automata
$c_1\colon X_1 \to \Fda (X_1)$ and $c_2\colon X_2 \to \Fda (X_2)$, the
distance between $(x, y)$ and $(x', y')$ in the product automaton
$\distFda \circ (c_1 \times c_2)$ is bounded by $\max(v_1, v_2)$ where
$v_1$ is the distance between $x, x'$ in $c_1$ and $v_2$ is the
distance between $y, y'$ in $c_2$.

The proof is to show that two conditions in Theorem~\ref{thm:suf_lift_dist}.
\ifarxiv
See Appendix~\ref{ap:sim_pseudo} for its proof.
\else
See \cite[Appendix A.3.2]{kori2024composing} for its proof.
\fi

\subsection{Bisimulation Metric for Markov Decision Processes} \label{sec:bisim_metric}
\emph{Bisimulation metric} is a quantitative notion of behavioral equivalences for probabilistic systems~\cite{DBLP:journals/tcs/DesharnaisGJP04, DBLP:journals/entcs/DengCPP06, DBLP:journals/tcs/BreugelW05}.
We shall see liftability of a distributive law for composing bisimulation metrics~\cite{DBLP:journals/tcs/DesharnaisGJP04}.

\paragraph{Codensity bisimulation data}
We first set-up coalgebraic formulation of Markov decision processes (MDPs) and bisimulation metric.
Let us write $\subdistmnd\colon \Set \to \Set$ for the probability distribution functor
which maps $X$ to the set of distributions on $X$.
Then MDPs are naturally modeled by $\pow\circ\subdistmnd$-coalgebras
$\coalg{\pow\subdistmnd}cX$. Hereafter we omit the composition between
$\pow$ and $\subdistmnd$.

For bisimulation metrics,
let us provide two codensity bisimulation data
for the Kantorovich and Hausdorff lifting.

The first codensity bisimulation data is
$(\pfibU p \PMet \Set 1 {\intv} {d_{\intv}}, \subdistmnd, e)$
where $e$ is the expectation function $e\colon \subdistmnd\intv\to \intv$.
The codensity lifting $\codlift{\subdistmnd}{d_{\intv}}{e}\colon \mathbf{PMet} \to \mathbf{PMet}$
maps $d \in \mathbf{PMet}_X$ to the Kantorovich pseudometric $\mathcal{K}(d)$ on $\subdistmnd X$:
\begin{displaymath}
  (\mu_1,\mu_2) \mapsto
  \sup_{k\in\PMet((X,d),(\intv,d_{\intv}))}
  |\Sigma_{x \in X} k(x) \cdot (\mu_1(x) - \mu_2(x))|.
\end{displaymath}

The second codensity bisimulation data is
$(\pfibU p \PMet \Set 1 {\intv} {d_{\intv}}, \pow, \mathrm{inf})$.
The codensity lifting $\codlift{\pow}{d_{\intv}}{\mathrm{inf}}\colon \mathbf{PMet} \to \mathbf{PMet}$
mapping $d \in \mathbf{PMet}_X$ to the Hausdorff $\mathcal{H}(d)$ pseudometric on $\pow X$:
\begin{align*}
  (A_1, A_2)
  \mapsto
  &\sup_{k\in\PMet((X,d),(\intv,d_{\intv}))}
    |\mathrm{inf}_{a_1 \in A_1} k(a_1) -  \mathrm{inf}_{a_2 \in A_2} k(a_2)| \\
  =& \max(\sup_{a_1 \in A_1}\inf_{a_2 \in A_2} d(a_1, a_2), \sup_{a_2 \in A_2}\inf_{a_1 \in A_1} d(a_1, a_2))
\end{align*}
Notice that $\mathrm{sup}\, \emptyset = 0$ and
$\mathrm{inf}\, \emptyset=1$, hence
$\mathcal{H}(d)(\emptyset, A_2) = 1$ for $A_2 \neq \emptyset$.  The
equality is proved in
\cite[Appendix 3]{DBLP:journals/ngc/KomoridaKHKHEH22}.

Then a
$\codlift{\pow}{d_{\intv}}{\mathrm{inf}}\circ\codlift{\subdistmnd}{d_{\intv}}{e}$-bisimulation
on $c$ is called a \emph{bisimulation metric},
and the bisimilarity is called the \emph{bisimilarity metric}~\cite{DBLP:journals/entcs/DengCPP06}.
It maps a pair of states $(x, y)$ to a distance that measures quantitative analogue to behavioral equivalences.

\paragraph{Binary one-step composition operation}
We have seen the binary one-step
composition $(\times, \distpow)$ for $\pow$ in \S{}\ref{eg:comp_kripke}. We next introduce the one for $\subdistmnd $
by
\[ \lambda^\subdistmnd (\mu, \mu') \coloneqq (x, y) \mapsto \mu(x)
  \cdot \mu'(y). \] This induces the binary one-step composition
$(\times, \lambda^{\powsub})$ for $\powsub$-coalgebras by
\[ \lambda^{\powsub} \defeq (\pow \circ \lambda^\subdistmnd) \bullet
   (\lambda^\pow \circ \subdistmnd^2) . \]

\paragraph{Modality for $\times$}
From the discussion above, it is sufficient to find a modality
$\sigma$ for $\times$ that lifts both
distributive laws $\lambda^\pow$ and $\lambda^\subdistmnd$. Then
both distributive laws $\lambda^\pow$ and $\lambda^\subdistmnd$ become
liftable, hence so is $\lambda^{\mathcal{P D}}$.  Among several
modalities introduced in \S{}\ref{sec:pmet_codlift}, only $\sigmant$
lifts both binary operations. Therefore:
\begin{proposition} \label{prop:local_comp_metric}
The distributive law $\lambda^{\powsub}$ is liftable on countable pseudometric spaces
    w.r.t.~$\Ncodlift{\times}{d_{\intv}}{\sigmant}{2} \circ (\codlift{\pow}{d_{\intv}}{\mathrm{inf}} \circ \codlift{\subdistmnd}{d_{\intv}}{e})^N$
    and $\codlift{\pow}{d_{\intv}}{\mathrm{inf}}\circ \codlift{\subdistmnd}{d_{\intv}}{e} \circ \Ncodlift{\times}{d_{\intv}}{\sigmant}{2}$.
  \qed
\end{proposition}

As a consequence of the proposition above, given two MDPs
$\coalg{\powsub}{c_i}{X_i}$ over at most countable states $X_i$
($i=1,2$), if $d_i$ is a bisimulation metric for $c_i$, then the
mapping
\begin{displaymath}
  d((x, y), (x', y')) \defeq \sigmant (d_1(x, x'), d_2(y, y'))
\end{displaymath}
becomes a bisimulation metric for the composite MDP
$\lambda^\powsub \circ(c_1 \times c_2)$.  Moreover, the above
proposition implies preservation of \emph{bisimilarity metrics}.  If
$d_i$ is the bisimilarity metric for $c_i$, then the bisimilarity
metric for the composite MDP $\lambda^\powsub \circ(c_1 \times c_2)$
is bounded by $d$.  This bound itself was shown in
  \cite{DBLP:journals/corr/GeblerLT16}.  We give a proof in terms of
  composition of distributive laws.

We prove Proposition~\ref{prop:local_comp_metric} by checking if
$\sigmant$ lifts both binary operations
$(\times, \lambda^\pow)$ and $(\times, \lambda^\subdistmnd)$. Other
$\sigma$ written in \S{}\ref{sec:pmet_codlift} do not satisfy all
conditions in the lemma.
\begin{lemma} \label{lem:pd_compat}
  Let $\sigma\colon \intv^2 \to \intv$ be a modality satisfying \eqref{eq:star}, so that $\Ncodlift{\times}{d_\intv}{\sigma}{2} = \times^\sigma$.
  \begin{enumerate}
  \item The modality $\sigma$ lifts $\lambda^{\pow}$ if $\sigma$ is
    concave, $\sigma$ preserves infimums, and
    $\sigma(1, x) = \sigma(x, 1) = 1$ for any $x \in \intv$.

  \item The modality $\sigma$ lifts
    $\lambda^{\subdistmnd}$ on countable
    pseudometric spaces if $\sigma$ is concave.
\qed
  \end{enumerate}
\end{lemma}

The second statement comes from the following lemma, which is a direct
adaptation of~\cite[Theorem~2.20]{DBLP:journals/corr/GeblerLT16} in the
context of pseudometric spaces.
It can be proved in almost the same way as the original theorem; here
we rely on the Kantorovich-Rubinstein duality theorem on at most
countable sets~\cite{villani2009optimal}.

\begin{lemma}
  \label{lem:gebler}
  Let $X, Y$ be at most countable sets,
$(X, d_1), (Y, d_2)  \in \PMet$,
  and $\sigma\colon  \intv^2 \to \intv$ be a concave function satisfying \eqref{eq:star}.
  Then
  for each
  $\mu_1, \mu_1' \in \subdistmnd X$ and
  $\mu_2, \mu_2' \in \subdistmnd Y$,
  \begin{equation} \label{eq:gebler}
    \mathcal{K}(d)(\lambda^{\subdistmnd}_{X, Y}(\mu_1, \mu_2), \lambda^{\subdistmnd}_{X, Y}(\mu_1', \mu_2')) \leq \sigma(\mathcal{K}(d_1)(\mu_1, \mu_1'), \mathcal{K}(d_2)(\mu_2, \mu_2'))
  \end{equation}
  where $d$ is the pseudometric on $X \times Y$ given by
  $d((x, y), (x', y')) \coloneqq \sigma(d_1(x, x'), d_2(y, y'))$.
  \qed
\end{lemma}

We did not apply Theorem~\ref{thm:suf_lift_dist} to this case
because
as of now,
it remains unsolved whether
the modality $\sigmant$ satisfies the second condition in Theorem~\ref{thm:suf_lift_dist}.
We note that the modality satisfies the first condition.
 \section{Composing Codensity Games} \label{sec:game}

We introduce a composition of codensity games, and show that it
preserves game invariants. As a consequence, we prove
the preservation of bisimilarities (the inequality shown in
Corollary~\ref{cor:presbisim})
under the sufficient condition in
Theorem~\ref{thm:suf_lift_dist}
by composition of games.

\subsection{Compositionality of Invariants}

We briefly recall definitions and results on codensity games.
We reserve the variable $i$ for $N$-indices,
and write $\vfami{x_i}{i}{N}$ for the sequence of mathematical entities $x_i$ (such as morphisms and objects)  indexed by $i \in N$.
\begin{definition}
  A \emph{safety game}
  is a game $\mathcal{G} = (Q_\duplicator, Q_\spoiler, E)$ played by two players $\duplicator$ (Duplicator) and $\spoiler$ (Spoiler)
  where
  $Q_\duplicator, Q_\spoiler$ are sets of positions of $\duplicator, \spoiler$ respectively
  and $E \subseteq (Q_\duplicator \times Q_\spoiler \cup Q_\spoiler \times Q_\duplicator)$
  is a set of possible moves.
  A play of $\mathcal{G}$ is a finite or infinite sequence of positions
  $q_0, q_1, \dots$
  such that $(q_i, q_{i+1}) \in E$ for each $i$.
  The player $\spoiler$ wins a play if the sequence is finite and the last position is in $Q_\duplicator$,
  and the player $\duplicator$ wins the game otherwise.

  A \emph{strategy} of $\duplicator$ is a partial function $s\colon Q^* \times Q_\duplicator \rightharpoonup Q_\spoiler$
  where $Q^*$ is the set of finite sequences of $Q_\duplicator \uplus Q_\spoiler$.
  A strategy $s$ of $\duplicator$ is \emph{winning} from $q$ if
  $\duplicator$ wins any play $q_0, q_1, \dots$ such that $q_0 = q$ and $q_{i+1} = s(q_0, q_1, \dots, q_{i})$ for each $q_i \in Q_\duplicator$.
  A position
  $q \in Q_\spoiler$ is \emph{winning} (for $\duplicator$) if there exists a winning strategy $s$ of $\duplicator$ from $q$.
\end{definition}

Winning positions on a safety games are characterizes by \emph{invariants}, which also induce winning strategies of $\duplicator$.
\begin{definition}
  A set $\mathcal{V} \subseteq Q_\spoiler$ is an \emph{invariant} for $\duplicator$ if
  for each $q \in \mathcal{V}$ and $q' \in Q_\duplicator$, $(q, q') \in E$ implies
  that
  there is $q'' \in \mathcal{V}$ such that $(q', q'') \in E$.
\end{definition}

\begin{proposition} \label{prop:win_invariant}
  A position $q \in Q_\spoiler$ is winning
  if and only if
  there is an invariant $\mathcal{V}$ for $\duplicator$ such that $q \in \mathcal{V}$.
  \qed
\end{proposition}

Codensity games are safety games that are induced by codensity liftings and coalgebras.
\begin{definition}[codensity games~\cite{DBLP:journals/ngc/KomoridaKHKHEH22}]
  Let $(\pfib p\EE\BB A\O\bO,F,\tau)$ be codensity bisimulation data
  and $c\colon X \to F(X)$ be an $F$-coalgebra.
The \emph{codensity game} $\cgame{c}$ is the safety game $(Q_\duplicator, Q_\spoiler, E)$ played by two players $\duplicator,\spoiler$
where $Q_\duplicator \coloneqq \{(a, k) \mid a \in A, k \in \BB(X, \O(a))\}$,
$Q_\spoiler \coloneqq \mathrm{Obj}(\EE_X)$, and $E \coloneqq {\to_\spoiler} \cup {\to_\duplicator}$ is
given by
\begin{align*}
&P \to_\spoiler (a, k) \text{ if }\tau_a \circ Fk \circ c\colon P\decent{\nrightarrow} \bO(a), \\
&(a, k) \to_\duplicator P' \text{ if } k\colon P' \decent{\nrightarrow} \bO(a).
\end{align*}
\end{definition}

We show characterizations of invariants and winning positions of codensity games, which are required in the latter part.
\begin{proposition}[{\cite{DBLP:journals/ngc/KomoridaKHKHEH22}}] \label{prop:win_bisim}
  \begin{enumerate}
    \item
  A set $\mathcal{V} \subseteq \mathrm{Obj}(\mathbb{E}_X)$ is an invariant for $\duplicator$ in the codensity game $\cgame{c}$
  if and only if $\bigsqcup_{P \in \mathcal{V}} P$ is an $\codlift{F}{\bO}{\tau}$-bisimulation on $c$.
    \item
  A position $P \in \EE_X$ is winning in $\mathcal{G}_c$
  if and only if
  $P$ is a \emph{witness} of codensity bisimilarity on $c$, that is, $P\sqsubseteq  \nu\bigl(c^* \circ \codlift{F}{\bO}{\tau} \bigr)$.
  \qed
  \end{enumerate}
\end{proposition}

Now we move to a codensity game of a composite coalgebra.
Given $F$-coalgebras $c_i\colon X_i \to F(X_i)$ ($i \in N$),
the codensity game $\cgame{\coalgmerge T\lambda{\vfami{c_i}{i}{N}}}$
of the composite coalgebra $\coalgmerge T\lambda{\vfami{c_i}{i}{N}}$
is given by
\begin{align*}
&(P \in \EE_{T(\vfami{X_i}{i}{N})}) \to_\spoiler (a, k)
 \text{ if }
\tau_a \circ Fk \circ (\lambda \circ T(\vfami{c_i}{i}{N}))\colon
      P \decent{\nrightarrow} \bO(a), \\
&(a \in A, k \in \BB(T(\vfami{X_i}{i}{N}), \O(a))) \to_\duplicator P' \text{ if $k\colon P' \decent{\nrightarrow} \bO(a)$.}
\end{align*}
In general, the codensity game $\cgame{\coalgmerge T\lambda{\vfami{c_i}{i}{N}}}$ may not be composed from each game of $\cgame{c_i}$:
in particular, positions and moves may not be composed by products, and therefore invariants may not be.

Instead, we define
a composite codensity game (Definition~\ref{def:comp_cgame}), where invariants are compositionally preserved:
\begin{definition}[composite codensity games] \label{def:comp_cgame}
  Let $(\pfib p\EE\BB A\O\bO,F,\tau)$ be a codensity bisimulation data
  and $c_i\colon X_i \to F(X_i)$ ($i \in N$) be $F$-coalgebras.  The
  \emph{composite codensity game} $\compcgame{c_i}{i}{N}{T}{\sigma}$
  is the safety game $(Q_\duplicator, Q_\spoiler, E)$ by two players
  $\duplicator,\spoiler$ where
  $Q_\duplicator \coloneqq \{(a, \vfami{k_i}{i}{N}) \mid a \in A,
  k_i\in \mathbb{C}(X_i,\Omega(a))\}$,
  $Q_\spoiler \coloneqq \prod_i \mathrm{Obj}(\mathbb{E}_{X_i})$,
  and $E \coloneqq \to_\spoiler \cup \to_\duplicator$ is given by
\begin{align*}
&\vfami{P_i}{i}{N} \to_\spoiler (a, \vfami{k_i}{i}{N})
 \text{ if }
\sigma_a \circ T(\vfami{\tau_a \circ Fk_i \circ c_i}{i}{N})\colon
      \Ncodlift{T}{\bO}{\sigma}{N}\vfami{P_i}{i}{N} \decent{\nrightarrow} \bO(a), \\
&(a, \vfami{k_i}{i}{N}) \to_\duplicator \vfami{P'_i}{i}{N} \text{ if there is $i \in N$ s.t.~$k_i\colon P'_i \decent{\nrightarrow} \bO(a)$.}
\end{align*}
\end{definition}

The following proposition
shows that we can construct an invariant by composing invariants of component games $\cgame{c_i}$.
\begin{proposition} \label{prop:invariant_compositional_bisim}
  If each $\mathcal{V}_i \subseteq \mathrm{Obj}(\mathbb{E}_{X_i})$ is an invariant of $\cgame{c_i}$ for $\duplicator$
  then $\prod_{i \in N}\mathcal{V}_i$ is an invariant of $\compcgame{c_i}{i}{N}{T}{\sigma}$ for $\duplicator$.
  \qed
\end{proposition}

\begin{example}
  Consider language equivalence of deterministic automata as presented in \S{}\ref{eg:lang_equiv}.
  For an $\Fda$-coalgebra $c\colon X \to \Fda X$,
  the codensity game $\cgame{c}$ is $((\Sigma \uplus \{\epsilon\}) \times \Set(X, 2), \mathrm{Obj}(\mathrm{EqRel}_X), \to_\spoiler \cup \to_\duplicator)$ with
  \begin{align*}
&P \to_\spoiler (\epsilon, k) \text{ if }\exists (x, y) \in P.~\pi_1 c(x) \neq \pi_1 c(y), \\
&P \to_\spoiler (a \in \Sigma, k) \text{ if }\exists (x, y) \in P.~k(\pi_2 c(x)) \neq k(\pi_2 c(y)), \\
&(a, k) \to_\duplicator P' \text{ if } \exists (x, y) \in P'.~k(x) \neq k(y),
  \end{align*}
and
  for two $\Fda$-coalgebras $\{c_i\colon X_i \to \Fda X_i\}_{i \in 2}$,
the composite codensity game $\compcgame{c_i}{i}{2}{\times}{\sigma_\land}$ is $((\Sigma \uplus \{\epsilon\}) \times \Set(X_1, 2) \times \Set(X_2, 2), \mathrm{Obj}(\mathrm{EqRel}_{X_1}) \times \mathrm{Obj}(\mathrm{EqRel}_{X_2}), \to_\spoiler \cup \to_\duplicator)$ with
  \begin{align*}
&(P_1, P_2) \to_\spoiler (\epsilon, k_1, k_2)  \\
&\quad \text{ if }\exists (x_1, y_1) \in P_1.~\exists (x_2, y_2) \in P_2. \\
&\phantom{\quad \text{ if }}\pi_1 c(x_1) \land \pi_1(c(x_2)) \neq \pi_1 c(y_1) \land \pi_1 c(y_2), \\
&(P_1, P_2) \to_\spoiler (a \in \Sigma, k_1, k_2) \\
&\quad \text{ if }\exists (x_1, y_1) \in P_1.~\exists (x_2, y_2) \in P_2. \\
&\phantom{\quad \text{ if }}k(\pi_2 c(x_1)) \land k(\pi_2(c(x_2))) \neq k(\pi_2 c(y_1)) \land k(\pi_2 c(y_2)), \\
&(a, k_1, k_2) \to_\duplicator (P'_1, P'_2) \text{ if } \exists i\in 2.~\exists (x, y) \in P'_i.~k(x_i) \neq k(y_i).
  \end{align*}
  In this composite codensity game,
  positions are composed by pairing
  and
  conditions are composed by the modality $\sigma_\land$.
\end{example}

\subsection{Preservation of Bisimilarities}

Under the sufficient condition in Theorem~\ref{thm:suf_lift_dist},
as discussed in \S{}\ref{sec:liftability},
we naturally have the following natural transformation above $\lambda$.
  \begin{align*}
  &R \circ \Sp[A]T\sigma \circ \Sp[A]{F^N}{\tau^N} \circ L_N
    \hfill \nonumber \\
  &\Rightarrow R \circ \Sp[A]F\tau \circ \Sp[A]T\sigma \circ L_N &\text{by \eqref{item:2-cell} in Thm.~\ref{thm:suf_lift_dist}}, \\
  &= \codlift{F}{\bO}{\tau} \Ncodlift{T}{\bO}{\sigma}{N} &\text{by \eqref{item:approximating} in Thm.~\ref{thm:suf_lift_dist}}. \\
  \end{align*}
It connects invariants in composite codensity games $\compcgame{c_i}{i}{N}{T}{\sigma}$,
which utilizes combination of modalities $\sigma \circ T(\vec{\tau})$,
and $\codlift{F}{\bO}{\tau}$-bisimulations composed by $\Ncodlift{T}{\bO}{\sigma}{N}$.
We introduce a
join-closure operation for invariants: for a set
$\mathcal{V} \subseteq \prod_{i \in N}
\mathrm{Obj}(\mathbb{E}_{X_i})$, we write $\overline{\mathcal{V}}$ for
$\mathcal{V} \cup \Big\{\fami{\bigsqcup_{\vfami{P_i}{i}{N} \in
    \mathcal{V}} P_i}{i}{N} \Big\}$.  We remark that if a set
$\mathcal{V}$ is an invariant of a codensity game $\mathcal{G}_c$ for
$\duplicator$ then so is $\overline{\mathcal{V}}$.

\begin{theorem} \label{thm:invariant_compositional_bisim}
  Assume that the sufficient condition in Theorem~\ref{thm:suf_lift_dist} holds
  and
  let $\mathcal{V} \subseteq \prod_{i \in N} \mathrm{Obj}(\mathbb{E}_{X_i})$
  be a set.
  \begin{enumerate}
    \item
  The set $\mathcal{V}$ is an invariant for $\duplicator$ in
  $\compcgame{c_i}{i}{N}{T}{\sigma}$
  if \\
  $\bigsqcup_{\vfami{P_i}{i}{N} \in \mathcal{V}}\Ncodlift{T}{\bO}{\sigma}{N}(\vfami{P_i}{i}{N})$
  is an $\codlift{F}{\bO}{\tau}$-bisimulation on $\coalgmerge T\lambda{\vfami{c_i}{i}{N}}$.
  \item
  The set $\overline{\mathcal{V}}$ is an invariant for $\duplicator$ in
$\compcgame{c_i}{i}{N}{T}{\sigma}$
  if and only if \\
  $\bigsqcup_{\vfami{P_i}{i}{N} \in \overline{\mathcal{V}}}\Ncodlift{T}{\bO}{\sigma}{N}(\vfami{P_i}{i}{N})$
  is an $\codlift{F}{\bO}{\tau}$-bisimulation on $\coalgmerge T\lambda{\vfami{c_i}{i}{N}}$.
  \qed
  \end{enumerate}
\end{theorem}
\begin{corollary} \label{cor:comp_win_invariant}
  Assume that the sufficient condition in Theorem~\ref{thm:suf_lift_dist} holds.
  If $\vfami{P_i}{i}{N}$ is a winning position for $\duplicator$
  with an invariant $\overline{\mathcal{V}}$
  then
  $\Ncodlift{T}{\bO}{\sigma}{N}(\vfami{P_i}{i}{N})$ is a witness  on $\coalgmerge T\lambda{\vfami{c_i}{i}{N}}$, that is, $\Ncodlift{T}{\bO}{\sigma}{N}(\vfami{P_i}{i}{N})\sqsubseteq  \nu\Bigl({\big(\coalgmerge T\lambda{\vfami{c_i}{i}{N}}\big)}^* \circ \codlift{F}{\bO}{\tau} \Bigr)$.
  \qed
\end{corollary}

Finally, through a composite codensity game,
we prove the preservation of bisimilarities (the inequalities shown in Corollary~\ref{cor:presbisim}).
In the proof, we instead prove the following equivalent statement:

\begin{proposition}
  Assume that the sufficient condition in Theorem~\ref{thm:suf_lift_dist} holds.
  If $P_i$ is a witness on $c_i$, that is, $P_i\sqsubseteq  \nu\Bigl({c_i}^* \circ \codlift{F}{\bO}{\tau} \Bigr)$  for each $i$,
  then $\Ncodlift{T}{\bO}{\sigma}{N}(\vfami{P_i}{i}{N})$ is a witness on $T_\lambda(c_1,\cdots,c_N)$, that is, $\Ncodlift{T}{\bO}{\sigma}{N}(\vfami{P_i}{i}{N}) \sqsubseteq \nu(T_\lambda(c_1,\cdots,c_N)^* \circ \codlift{F}{\bO}{\tau})$.
\end{proposition}
We outline the alternative proof of the proposition by composite codensity games.
If each $P_i$ is a witness on codensity bisimilarity on $c_i$,
$P_i$ is a winning position for $\duplicator$ in $\cgame{c_i}$ by Proposition~\ref{prop:win_bisim}.
Then
Proposition~\ref{prop:win_invariant} ensures the existence of
invariants $\overline{\mathcal{V}}_i$ of $\cgame{c_i}$ such that $P_i \in \overline{\mathcal{V}}_i$.
Proposition~\ref{prop:invariant_compositional_bisim}
and
$\prod_{i \in N} \overline{\mathcal{V}}_i = \overline{\prod_{i \in N} \mathcal{V}_i}$
imply that $\vfami{P_i}{i}{N}$
is a winning position with the invariant $\overline{\prod_{i \in N} \mathcal{V}_i}$.
Therefore, by
Corollary~\ref{cor:comp_win_invariant} we conclude
that $\Ncodlift{T}{\bO}{\sigma}{N}(\vfami{P_i}{i}{N})$ is a witness of codensity bisimilarity on $\coalgmerge T\lambda{\vfami{c_i}{i}{N}}$.

\section{Conclusion}
In this paper, we have presented generalized codensity liftings in a 2-categorical framework. This development has facilitated the derivation of liftings of structure functors, especially binary product functors, and enhanced our understanding and manipulation of codensity liftings.

By integrating the 2-categorical framework of generalized codensity liftings and
structure functors through codensity liftings, we were able to identify a sufficient condition to lift a distributive law between functors into one between codensity liftings.
Additionally, we have investigated our sufficient condition from the perspective of compositional reasoning of codensity games.

For future work,
we will continue to explore and expand the scope of our approach, particularly considering
different structure functors.

  Additionally, we intend to streamline the current categorical notation using the syntax of categorical logic.

Recently, higher-order abstract GSOS frameworks~\cite{DBLP:journals/pacmpl/GoncharovMSTU23,DBLP:conf/lics/Urbat00MS23,DBLP:conf/fossacs/GoncharovSSTU24} are actively studied:
they extend behavior endofunctors to bifunctors and use certain dinatural transformations for distributive laws.
Seeking a higher-order extension of our framework, that is,
studying a liftability of distributive laws along codensity liftings in a higher-order setting is an interesting future direction.
More generally, in the current paper we focused on compositionality for individual operations. Accomodating
languages and calculi in our approach as defined, for instance, by (abstract) GSOS specifications, remains future work.

\section*{Acknowledgement}
  The authors are gradeful to Ichiro Hasuo, and Noboru Isobe for helpful discussions,
  and to the anonymous reviewers for helpful feedback.
  MK, KW and SK were supported by ERATO HASUO Metamathematics for
  Systems Design Project (No. JPMJER1603), JST.
  MK was supported by JSPS DC KAKENHI Grant (No. 22J21742).
  KW was supported by the JST grants No.~JPMJFS2136 and  JPMJAX23CU.
  JR was supported by the NWO grant No.~OCENW.M20.053.

\bibliographystyle{plain}
\bibliography{lics2024}

\begin{thebibliography}{10}

\bibitem{DBLP:conf/fsttcs/BaldanBKK14}
Paolo Baldan, Filippo Bonchi, Henning Kerstan, and Barbara K{\"{o}}nig.
\newblock Behavioral metrics via functor lifting.
\newblock In {\em {FSTTCS}}, volume~29 of {\em LIPIcs}, pages 403--415. Schloss Dagstuhl - Leibniz-Zentrum f{\"{u}}r Informatik, 2014.

\bibitem{DBLP:journals/lmcs/BaldanBKK18}
Paolo Baldan, Filippo Bonchi, Henning Kerstan, and Barbara K{\"{o}}nig.
\newblock Coalgebraic behavioral metrics.
\newblock {\em Log. Methods Comput. Sci.}, 14(3), 2018.

\bibitem{DBLP:conf/stacs/BeoharG0MFSW24}
Harsh Beohar, Sebastian Gurke, Barbara K{\"{o}}nig, Karla Messing, Jonas Forster, Lutz Schr{\"{o}}der, and Paul Wild.
\newblock Expressive quantale-valued logics for coalgebras: An adjunction-based approach.
\newblock In {\em {STACS}}, volume 289 of {\em LIPIcs}, pages 10:1--10:19. Schloss Dagstuhl - Leibniz-Zentrum f{\"{u}}r Informatik, 2024.

\bibitem{DBLP:journals/jacm/BloomIM95}
Bard Bloom, Sorin Istrail, and Albert~R. Meyer.
\newblock Bisimulation can't be traced.
\newblock {\em J. {ACM}}, 42(1):232--268, 1995.

\bibitem{DBLP:conf/concur/Bonchi0P18}
Filippo Bonchi, Barbara K{\"{o}}nig, and Daniela Petrisan.
\newblock Up-to techniques for behavioural metrics via fibrations.
\newblock In {\em {CONCUR}}, volume 118 of {\em LIPIcs}, pages 17:1--17:17. Schloss Dagstuhl - Leibniz-Zentrum f{\"{u}}r Informatik, 2018.

\bibitem{DBLP:journals/acta/BonchiPPR17}
Filippo Bonchi, Daniela Petrisan, Damien Pous, and Jurriaan Rot.
\newblock A general account of coinduction up-to.
\newblock {\em Acta Informatica}, 54(2):127--190, 2017.

\bibitem{DBLP:journals/entcs/DengCPP06}
Yuxin Deng, Tom Chothia, Catuscia Palamidessi, and Jun Pang.
\newblock Metrics for action-labelled quantitative transition systems.
\newblock In {\em {QAPL}}, volume 153 of {\em Electronic Notes in Theoretical Computer Science}, pages 79--96. Elsevier, 2005.

\bibitem{DBLP:journals/tcs/DesharnaisGJP04}
Jos{\'{e}}e Desharnais, Vineet Gupta, Radha Jagadeesan, and Prakash Panangaden.
\newblock Metrics for labelled markov processes.
\newblock {\em Theor. Comput. Sci.}, 318(3):323--354, 2004.

\bibitem{DBLP:conf/qest/DesharnaisLT08}
Jos{\'{e}}e Desharnais, Fran{\c{c}}ois Laviolette, and Mathieu Tracol.
\newblock Approximate analysis of probabilistic processes: Logic, simulation and games.
\newblock In {\em {QEST}}, pages 264--273. {IEEE} Computer Society, 2008.

\bibitem{DBLP:conf/icalp/FijalkowKP17}
Nathana{\"{e}}l Fijalkow, Bartek Klin, and Prakash Panangaden.
\newblock Expressiveness of probabilistic modal logics, revisited.
\newblock In {\em {ICALP}}, volume~80 of {\em LIPIcs}, pages 105:1--105:12. Schloss Dagstuhl - Leibniz-Zentrum f{\"{u}}r Informatik, 2017.

\bibitem{DBLP:conf/lics/FordMSB022}
Chase Ford, Stefan Milius, Lutz Schr{\"{o}}der, Harsh Beohar, and Barbara K{\"{o}}nig.
\newblock Graded monads and behavioural equivalence games.
\newblock In {\em {LICS}}, pages 61:1--61:13. {ACM}, 2022.

\bibitem{DBLP:journals/corr/GeblerLT16}
Daniel Gebler, Kim~G. Larsen, and Simone Tini.
\newblock Compositional bisimulation metric reasoning with probabilistic process calculi.
\newblock {\em Log. Methods Comput. Sci.}, 12(4), 2016.

\bibitem{DBLP:conf/fossacs/GoncharovHNSW23}
Sergey Goncharov, Dirk Hofmann, Pedro Nora, Lutz Schr{\"{o}}der, and Paul Wild.
\newblock Kantorovich functors and characteristic logics for behavioural distances.
\newblock In Orna Kupferman and Pawel Sobocinski, editors, {\em Foundations of Software Science and Computation Structures - 26th International Conference, FoSSaCS 2023, Held as Part of the European Joint Conferences on Theory and Practice of Software, {ETAPS} 2023, Paris, France, April 22-27, 2023, Proceedings}, volume 13992 of {\em Lecture Notes in Computer Science}, pages 46--67. Springer, 2023.

\bibitem{DBLP:journals/pacmpl/GoncharovMSTU23}
Sergey Goncharov, Stefan Milius, Lutz Schr{\"{o}}der, Stelios Tsampas, and Henning Urbat.
\newblock Towards a higher-order mathematical operational semantics.
\newblock {\em Proc. {ACM} Program. Lang.}, 7({POPL}):632--658, 2023.

\bibitem{DBLP:conf/fossacs/GoncharovSSTU24}
Sergey Goncharov, Alessio Santamaria, Lutz Schr{\"{o}}der, Stelios Tsampas, and Henning Urbat.
\newblock Logical predicates in higher-order mathematical operational semantics.
\newblock In {\em FoSSaCS {(2)}}, volume 14575 of {\em Lecture Notes in Computer Science}, pages 47--69. Springer, 2024.

\bibitem{DBLP:journals/mscs/HasuoKC18}
Ichiro Hasuo, Toshiki Kataoka, and Kenta Cho.
\newblock Coinductive predicates and final sequences in a fibration.
\newblock {\em Math. Struct. Comput. Sci.}, 28(4):562--611, 2018.

\bibitem{DBLP:journals/iandc/HermidaJ98}
Claudio Hermida and Bart Jacobs.
\newblock Structural induction and coinduction in a fibrational setting.
\newblock {\em Inf. Comput.}, 145(2):107--152, 1998.

\bibitem{herr74:topo}
Horst Herrlich.
\newblock Topological functors.
\newblock {\em General Topology and its Applications}, 4(2):125 -- 142, 1974.

\bibitem{DBLP:books/daglib/0023251}
Bart Jacobs.
\newblock {\em Categorical Logic and Type Theory}, volume 141 of {\em Studies in logic and the foundations of mathematics}.
\newblock North-Holland, 2001.

\bibitem{DBLP:books/cu/J2016}
Bart Jacobs.
\newblock {\em Introduction to Coalgebra: Towards Mathematics of States and Observation}, volume~59 of {\em Cambridge Tracts in Theoretical Computer Science}.
\newblock Cambridge University Press, 2016.

\bibitem{DBLP:conf/calco/KatsumataS15}
Shin{-}ya Katsumata and Tetsuya Sato.
\newblock Codensity liftings of monads.
\newblock In {\em {CALCO}}, volume~35 of {\em LIPIcs}, pages 156--170. Schloss Dagstuhl - Leibniz-Zentrum f{\"{u}}r Informatik, 2015.

\bibitem{DBLP:journals/lmcs/KatsumataSU18}
Shin{-}ya Katsumata, Tetsuya Sato, and Tarmo Uustalu.
\newblock Codensity lifting of monads and its dual.
\newblock {\em Log. Methods Comput. Sci.}, 14(4), 2018.

\bibitem{DBLP:journals/ngc/KomoridaKHKHEH22}
Yuichi Komorida, Shin{-}ya Katsumata, Nick Hu, Bartek Klin, Samuel Humeau, Clovis Eberhart, and Ichiro Hasuo.
\newblock Codensity games for bisimilarity.
\newblock {\em New Gener. Comput.}, 40(2):403--465, 2022.

\bibitem{DBLP:conf/lics/KomoridaKKRH21}
Yuichi Komorida, Shin{-}ya Katsumata, Clemens Kupke, Jurriaan Rot, and Ichiro Hasuo.
\newblock Expressivity of quantitative modal logics : Categorical foundations via codensity and approximation.
\newblock In {\em {LICS}}, pages 1--14. {IEEE}, 2021.

\bibitem{DBLP:conf/concur/KonigM18}
Barbara K{\"{o}}nig and Christina Mika{-}Michalski.
\newblock (metric) bisimulation games and real-valued modal logics for coalgebras.
\newblock In {\em {CONCUR}}, volume 118 of {\em LIPIcs}, pages 37:1--37:17. Schloss Dagstuhl - Leibniz-Zentrum f{\"{u}}r Informatik, 2018.

\bibitem{DBLP:conf/concur/Lago023}
Ugo~Dal Lago and Maurizio Murgia.
\newblock Contextual behavioural metrics.
\newblock In {\em {CONCUR}}, volume 279 of {\em LIPIcs}, pages 38:1--38:17. Schloss Dagstuhl - Leibniz-Zentrum f{\"{u}}r Informatik, 2023.

\bibitem{DBLP:journals/tcs/MousaviRG07}
Mohammad~Reza Mousavi, Michel~A. Reniers, and Jan~Friso Groote.
\newblock {SOS} formats and meta-theory: 20 years after.
\newblock {\em Theor. Comput. Sci.}, 373(3):238--272, 2007.

\bibitem{DBLP:conf/cmcs/SprungerKDH18}
David Sprunger, Shin{-}ya Katsumata, J{\'{e}}r{\'{e}}my Dubut, and Ichiro Hasuo.
\newblock Fibrational bisimulations and quantitative reasoning.
\newblock In {\em {CMCS}}, volume 11202 of {\em Lecture Notes in Computer Science}, pages 190--213. Springer, 2018.

\bibitem{DBLP:journals/logcom/SprungerKDH21}
David Sprunger, Shin{-}ya Katsumata, J{\'{e}}r{\'{e}}my Dubut, and Ichiro Hasuo.
\newblock Fibrational bisimulations and quantitative reasoning: Extended version.
\newblock {\em J. Log. Comput.}, 31(6):1526--1559, 2021.

\bibitem{DBLP:journals/igpl/Stirling99}
Colin Stirling.
\newblock Bisimulation, modal logic and model checking games.
\newblock {\em Log. J. {IGPL}}, 7(1):103--124, 1999.

\bibitem{DBLP:conf/lics/TuriP97}
Daniele Turi and Gordon~D. Plotkin.
\newblock Towards a mathematical operational semantics.
\newblock In {\em {LICS}}, pages 280--291. {IEEE} Computer Society, 1997.

\bibitem{DBLP:conf/calco/TurkenburgBKR23}
Ruben Turkenburg, Harsh Beohar, Clemens Kupke, and Jurriaan Rot.
\newblock Forward and backward steps in a fibration.
\newblock In {\em {CALCO}}, volume 270 of {\em LIPIcs}, pages 6:1--6:18. Schloss Dagstuhl - Leibniz-Zentrum f{\"{u}}r Informatik, 2023.

\bibitem{DBLP:conf/lics/Urbat00MS23}
Henning Urbat, Stelios Tsampas, Sergey Goncharov, Stefan Milius, and Lutz Schr{\"{o}}der.
\newblock Weak similarity in higher-order mathematical operational semantics.
\newblock In {\em {LICS}}, pages 1--13, 2023.

\bibitem{franck}
Franck van Breugel.
\newblock The metric monad for probabilistic nondeterminism.
\newblock http://www.cse.yorku.ca/~franck/research/drafts/monad.pdf, 2005.

\bibitem{DBLP:journals/tcs/BreugelW05}
Franck van Breugel and James Worrell.
\newblock A behavioural pseudometric for probabilistic transition systems.
\newblock {\em Theor. Comput. Sci.}, 331(1):115--142, 2005.

\bibitem{villani2009optimal}
C{\'e}dric Villani.
\newblock {\em Optimal transport: old and new}, volume 338.
\newblock Springer, 2009.

\end{thebibliography}

\ifarxiv
\clearpage
\onecolumn
\appendix
\section{Omitted Proofs}

\subsection{Omitted Proofs for Section~\ref{sec:codlift_structure}}
\begin{proposition}[Proof for Example.~\ref{eg:ncod_times}]
  In Example.~\ref{eg:ncod_times},
  the \ptfibs{}
  $\pfibU p {\mathbf{EqRel}} \Set 1 2 {(2, \mathrm{Eq_2})}$
  and
  $\pfibU p {\mathbf{Top}} \Set 1 2 {(2, \{\emptyset, \{\mathrm{true}\}, 2\})}$
  satisfy
  $\conc{p}{\bO} \abs{p}{\bO} = \mathrm{id}$.
\end{proposition}
\begin{proof}
  $\conc{p}{\bO} \abs{p}{\bO} = \mathrm{id}$ is equivalent to $\bigsqcap_{k\colon \bO \to \bO}k^*\bO = \bO$.
  \begin{itemize}
    \item ($\pfibU p {\mathbf{EqRel}} \Set 1 2 {(2, \mathrm{Eq_2})}$):
      $\bigsqcap_{k\colon \bO \to \bO}k^*\bO = \{(x, y) \mid \forall k\colon \bO \to \bO.~kx = ky\} = \bO$ holds.
      The last equality is because $(x, y)$ in l.h.s.~satisfies $x = y$ by letting $k \coloneqq \mathrm{id}$, and $(x, y)$ in r.h.s.~satisfies $kx = ky$ for each $k\colon \bO \to \bO$.
    \item ($\pfibU p {\mathbf{Top}} \Set 1 2 {(2, \{\emptyset, \{\mathrm{true}\}, 2\})}$):
    $\bigsqcap_{k\colon \bO \to \bO}k^*\bO$ is the topology generated by $\bigcup_{k\colon \bO \to \bO} \{k^{-1}A \mid A \in \bO\}$.
    $\bO$ is included in the topology by $k \coloneqq \mathrm{id}\colon \bO \to \bO$,
    and the topology is included in $\bO$ because $k^{-1}A \in \bO$ for each $k\colon \bO \to \bO$ and $A \in \bO$.
  \end{itemize}
\end{proof}
\begin{proof}[Proof of Proposition~\ref{prop:meet-free_pmet}]
  \begin{enumerate}
    \item ($\Rightarrow$) Suppose $\times^\sigma\colon \mathbf{PMet}^2 \to \mathbf{PMet}$ is a lifting of the product functor $\times$.
    Consider arbitrary $a, b, c, d \in [0, 1]$.
    For $x, x', x'' \in \{a, b, c, d\}$
    we write $P_x \in \mathbf{PMet}_2$ for the pseudometric satisfying $P_x(0, 1) = x$,
    and $P_{x, x', x''} \in \mathbf{PMet}_3$ (when $x+x' \leq x'')$ for the one mapping $(0, 1)$ to $x$, $(1, 2)$ to $x'$, $(0, 2)$ to $x+x'$
    where $2=\{0, 1\}$ and $3=\{0, 1, 2\}$.
    The modality $\sigma$ is monotone because
    $a \leq c$ and $b \leq d$ imply
    $\sigma(a, b) = \sigma(P_a(0, 1), P_b(0, 1)) \leq \sigma(P_{c}(0, 1), P_{d}(0, 1)) = \sigma(c, d)$
    by monotonicity of the functor $\lifttimes{\sigma}$.
    $\sigma(0, 0) = \lifttimes{\sigma}(P_a, P_a)((0, 0), (0, 0)) = 0$ because $\lifttimes{\sigma}(P_a, P_a)$ is a pseudometric.
    $\sigma(a, b) = \lifttimes{\sigma}(P_{c, |a-c|, a}, P_{d, |b-d|, b})((0, 0), (2, 2)) \leq \lifttimes{\sigma}(P_{c, |a-c|, a}, P_{d, |b-d|, b})((0, 0), (1, 1)) + \lifttimes{\sigma}(P_{c, |a-c|, a}, P_{d, |b-d|, b})((1, 1), (2, 2)) = \sigma(|a-c|, |b-d|)+\sigma(c, d)$.
    Thus we have $\sigma(a, b) - \sigma(c, d) \leq \sigma(|a-c|, |b-d|)$.

    ($\Leftarrow$) Suppose $\sigma$ satisfies \eqref{eq:star}.
    Then $\lifttimes{\sigma}(P, Q)$ is a pseudometric for each $P, Q \in \mathbf{PMet}$:
    For arbitrary $x, x', x'' \in pP$, $y, y', y'' \in pQ$,
    \begin{itemize}
      \item
    $\lifttimes{\sigma}(P, Q)((x, y), (x, y)) = \sigma(P(x, x), Q(y, y)) = \sigma(0, 0) = 0$,
    \item
    $\lifttimes{\sigma}(P, Q)((x, y), (x', y')) = \sigma(P(x, x'), Q(y, y')) =  \sigma(P(x', x), Q(y', y)) = \lifttimes{\sigma}(P, Q)((x', y'), (x, y))$,
    \item
    $\lifttimes{\sigma}(P, Q)((x, y), (x'', y'')) = \sigma(P(x, x''), Q(y, y'')) \leq \sigma(P(x, x')+P(x', x''), Q(y, y')+Q(y', y'')) \leq \sigma(P(x, x'), Q(y, y'))+\sigma(P(x', x''), Q(y', y'')) = \lifttimes{\sigma}(P, Q)((x, y), (x', y')) + \lifttimes{\sigma}(P, Q)((x', y'), (x'', y''))$.
    \end{itemize}
    For each $f\colon P \to P'$ and $g\colon Q \to Q'$ in $\mathbf{PMet}$,
    $f\times g \colon \lifttimes{\sigma}(P, Q) \decent{\rightarrow} \lifttimes{\sigma}(P', Q')$ by monotonicity of $\sigma$.

    \item The third condition of \eqref{eq:star} is equivalent to
    $|\sigma(a, b) - \sigma(c, d)| \leq \sigma(|a-c|, |b-d|)$ for each $a, b, c, d \in [0, 1]$,
    and it corresponds to $\sigma\colon \times^\sigma(d_{[0, 1]}) \decent{\rightarrow} d_{[0, 1]}$.
    Therefore it induces $\times^\sigma(P, Q) \sqsubseteq \Ncodlift{\times}{d_{[0, 1]}}{\sigma}{2}(P, Q)$ for each $P, Q \in \mathbf{PMet}$ by Proposition~\ref{prop:codlift_decent}.
    For each $P, Q \in \mathbf{PMet}$,
    $\times^\sigma(P, Q) \sqsupseteq \Ncodlift{\times}{d_{[0, 1]}}{\sigma}{2}(P, Q)$
    because for each $x, x' \in pP$ and $y, y' \in pQ$,
    letting $k_1(t) \coloneqq P(x, t)$ and $k_2(t) \coloneqq Q(y, t)$,
    \begin{align*}
    \times^\sigma(P, Q)((x, y), (x', y')) &= \sigma(P(x, x'), Q(y, y')) \\
    &= \sigma(k_1(x'), k_2(y')) \\
    &= |\sigma(k_1(x), k_2(y)) - \sigma(k_1(x'), k_2(y'))| \\
    &= (\sigma \circ k_1 \times k_2)^*d_{[0, 1]}((x, y), (x', y')) \\
    &\sqsupseteq \Ncodlift{\times}{d_{[0, 1]}}{\sigma}{2}(P, Q)((x, y), (x', y')).
    \end{align*}
  \end{enumerate}
\end{proof}

\begin{proposition} \label{ap:prop:sigma_negneg}
  The modalities $\sigmant, \sigma_\lor$ defined by $\sigmant(a, b) \coloneqq 1-(1-a)(1-b)$ and $\sigma_\lor(a, b) \coloneqq \max(a, b)$ satisfy \eqref{eq:star}.
\end{proposition}
\begin{proof}
  For each $\sigma \in \{\sigmant, \sigma_\lor\}$,
  It is easy to show that $\sigma$ is monotone  and $\sigma(0, 0)=0$.
  Here we only show $\sigma(a,b) - \sigma(c, d) \leq \sigma(|a-c|, |b-d|)$
  for each $a,b,c,d \in [0, 1]$.

  \begin{itemize}
    \item $(\sigma = \sigmant)$:
  Let $v_1 \coloneqq a$, $\epsilon_1 \coloneqq a-c$,
  $v_2 \coloneqq b$, and $\epsilon_2 \coloneqq b-d$.
  Then
  $\max(0, \epsilon_i) \leq v_i \leq \min(1, 1+\epsilon_i) \, (i=1, 2)$ and
(lhs) $= \sigmant(v_1, v_2) - \sigmant(v_1-\epsilon_1, v_2-\epsilon_2)
  = \epsilon_1(1-v_2)+\epsilon_2(1-v_1)+\epsilon_1\epsilon_2$.

  If $\epsilon_1, \epsilon_2 \geq 0$,
  (lhs) $\leq \epsilon_1(1-\epsilon_2)+\epsilon_2(1-\epsilon_1)+\epsilon_1\epsilon_2
  = \sigmant(\epsilon_1, \epsilon_2) \leq$ (rhs).

  If $\epsilon_1 \geq 0$ and $\epsilon_2 \leq 0$,
  (lhs) $\leq \epsilon_1(1-\epsilon_2)+\epsilon_2(1-1)+\epsilon_1\epsilon_2
  = \epsilon_1 \leq$ (rhs).

  If $\epsilon_1 \leq 0$ and $\epsilon_2 \geq 0$,
  (lhs) $\leq \epsilon_1(1-1)+\epsilon_2(1-\epsilon_1)+\epsilon_1\epsilon_2
  = \epsilon_2 \leq$ (rhs).

  If $\epsilon_1 \leq 0$ and $\epsilon_2 \leq 0$,
  (lhs) $\leq \epsilon_1(1-(1+\epsilon_2))+\epsilon_2(1-(1+\epsilon_1))+\epsilon_1\epsilon_2
  = -\epsilon_1\epsilon_2 \leq 0 \leq$ (rhs).
    \item $(\sigma = \sigma_{\lor})$:
      Since $\sigma_\lor(x, y) = \sigma_\lor(y, x)$ for each $x, y \in [0, 1]$,
      we can assume $a \leq b$ without loss of generality.

      If $c \leq d$, (lhs) = $b - d \leq |b-d| \leq$ (rhs).

      If $c > d$, (lhs) = $b - c$. If $b \leq c$, (lhs) $\leq 0 \leq$ (rhs).
      Otherwise ($b \geq c$), $b - c \leq b - d$ because $d \leq c \leq b$.
      Therefore (lhs) $= b - c \leq |b - d| \leq $ (rhs).
  \end{itemize}
\end{proof}

\begin{proof}[Proof of Proposition~\ref{prop:isom_f}]
  For each $\fami{P_i}{i}{N} \in \EE^N$ and $a \in A$,
$\{\sigma \circ T(\fami{k_i}{i}{N}) \mid k_i \colon P_i \decent{\to} \bO(a) \}
= \{\sigma \circ T(\fami{f \circ k_i}{i}{N}) \mid k_i \colon P_i \decent{\to} \bO(a) \}$
because
for each $\fami{k_i\colon P_i \decent{\to} \bO(a)}{i}{N}$, $\fami{f \circ k_i}{i}{N}, \fami{f^{-1} \circ k_i}{i}{N}\colon P_i \decent{\to} \bO(a)$ by $f^*\bO = \bO = (f^{-1})^*\bO$.
It induces
$\Sp[A]{T}\sigma \circ \abs p {\cptpl\bO N}
= \Sp[A]{T}{\sigma \circ T(f^N)} \circ \abs p {\cptpl\bO N}$.
Therefore, $\Ncodlift{T}{\bO}{\sigma}{N} =
  \conc p {\cptpl\bO N} \circ \Sp[A]{T}{\sigma} \circ \abs p {\cptpl\bO N}
  =
  \conc p {\cptpl\bO N} \circ \Sp[A]{T}{\sigma \circ T(f^N)} \circ \abs p {\cptpl\bO N}
  =
\Ncodlift{T}{\bO}{\sigma \circ T(f^N)}{N}$.

Similarly,
  for each $S \in \Sp[A]{\BB^N}{\bO}$ and $a \in A$,
$\bigsqcap_{a \in A, \fami{k_i}{i}{N} \in S} (\sigma \circ T(\fami{k_i}{i}{N}))^*\bO(a)
= \bigsqcap_{a \in A, \fami{k_i}{i}{N} \in S} (\sigma \circ T(\fami{k_i}{i}{N}))^*f^*\bO(a)
= \bigsqcap_{a \in A, \fami{k_i}{i}{N} \in S} (f \circ \sigma \circ T(\fami{k_i}{i}{N}))^*\bO(a)$.
It induces $\conc p {\cptpl\bO N} \circ \Sp[A]{T}{\sigma} = \conc p {\cptpl\bO N} \circ \Sp[A]{T}{f \circ \sigma}$.
Therefore,
$\Ncodlift{T}{\bO}{\sigma}{N} =
  \conc p {\cptpl\bO N} \circ \Sp[A]{T}{\sigma} \circ \abs p {\cptpl\bO N}
  =
  \conc p {\cptpl\bO N} \circ \Sp[A]{T}{f \circ \sigma} \circ \abs p {\cptpl\bO N}
  =
\Ncodlift{T}{\bO}{f \circ \sigma}{N}$.
\end{proof}

\subsection{Omitted Proofs for Section~\ref{sec:liftability}}
\begin{lemma} \label{lem:ap_join}
  In the setting of Proposition~\ref{prop:lift_lambda_join},
  if the three conditions in the proposition hold, then for each $a \in A$,
  \begin{displaymath}
  \conc p{\bO(a)} \circ \Sp[1]F{\tau_a}(S_a) \sqsubseteq \conc p{\bO(a)} \circ \Sp[1]F{\tau_a}\circ \abs{p}{\bO(a)} \circ \conc{p}{\bO}(S).
  \end{displaymath}
\end{lemma}
\begin{proof}
  For each $a \in A$ and $k\colon \bigsqcap_{a' \in A, k' \in S_{a'}}k'^*\mathbf{\Omega}(a') \to \mathbf{\Omega}(a)$,
    \begin{align*}
    \bigsqcap_{k' \in S_a} (\tau_a \circ Fk')^*\bO(a)
    &\sqsubseteq \bigsqcap_{k' \in S'_k} (\tau_a \circ Fk')^* \bO(a) &\text{since }S'_k \subseteq S_a \\
    &\sqsubseteq (\bigvee_{k' \in S'_k} \tau_a \circ Fk')^* \bO(a) &\text{by the second condition} \\
    &= (\tau_a \circ F(\bigvee_{k' \in S'_k} k'))^* \bO(a) &\text{by the first condition} \\
    &= (\tau_a \circ Fpk)^* \bO(a).&\text{by the third condition}
    \end{align*}
\end{proof}

\begin{proof}[Proof of Proposition~\ref{prop:lift_lambda_join}]
  Because $\conc p\bO \circ \Sp[A]F{\tau}(S') = \bigsqcap_{a \in A}\conc p{\bO(a)} \circ \Sp[1]F{\tau_a} (S'_a)$ for each $S' \in \Sp[A]{\BB}{\O}$,
  $S$ is approximating to $\codlift{F}{\bO}{\tau}$ is equivalent to
  $\conc p\bO \circ \Sp[A]F{\tau}(S) \sqsubseteq \conc p{\bO(a)} \circ \Sp[1]F{\tau_a}\circ \abs{p}{\bO(a)} \circ \conc{p}{\bO}(S)$ for each $a \in A$.
  Lemma~\ref{lem:ap_join} and $\conc p\bO \circ \Sp[A]F{\tau}(S) \sqsubseteq \conc p{\bO(a)} \circ \Sp[1]F{\tau_a}(S_a)$ for each $a \in A$ conclude the proof.
\end{proof}

\begin{proof}[Proof of Proposition~\ref{prop:transfer_principle_decent}]
	Suppose that $\lambda$ is liftable with respect to $\Ncodlift{T}{\bO}{\sigma}{N}$, $\codlift{F}{\bO}{\tau}$, i.e.~there is a natural transformation $\dot{\lambda}\colon \Ncodlift{T}{\bO}{\sigma}{N}(\codlift{F}{\bO}{\tau})^N \Rightarrow \codlift{F}{\bO}{\tau} \Ncodlift{T}{\bO}{\sigma}{N}$ above $\lambda$.
	Then we have the natural transformation $G \circ \dot{\lambda}$ above $\lambda$.
	The natural transformation $G \circ \dot{\lambda}$ is
	from
	$\Ncodlift{T}{G\bO}{\sigma}{N}\circ (\codlift{F}{G\bO}{\tau})^N\circ G^N$ to $\codlift{F}{G\bO}{\tau} \circ \Ncodlift{T}{G\bO}{\sigma}{N} \circ G^N$
  because
	$G\circ \Ncodlift{T}{\bO}{\sigma}{N}\circ (\codlift{F}{\bO}{\tau})^N = \Ncodlift{T}{G\bO}{\sigma}{N}\circ (\codlift{F}{G\bO}{\tau})^N\circ G^N$
	and $G \circ \codlift{F}{\bO}{\tau} \circ \Ncodlift{T}{\bO}{\sigma}{N}
	= \codlift{F}{G\bO}{\tau} \circ \Ncodlift{T}{G\bO}{\sigma}{N} \circ G^N$ hold
	by the discussion of \S{}\ref{sec:transfer}.

	Because $G$ is full, it ensures the existence of a natural transformation from $\Ncodlift{T}{G\bO}{\sigma}{N}(\codlift{F}{G\bO}{\tau})^N$ to $\codlift{F}{G\bO}{\tau} \Ncodlift{T}{G\bO}{\sigma}{N}$ above $\lambda$.
\end{proof}

\begin{proof}[Proof of Proposition~\ref{prop:lift_dist_a}]
  $\Ncodlift{T}{\bO}{\sigma}{N}\circ(\codlift{F}{\bO}{\tau})^N
  \Rightarrow \bigsqcap_{a \in A} \Ncodlift{T}{\bO}{\sigma}{N}\circ(\codlift{F}{\bO(a)}{\tau_a})^N
  \Rightarrow \bigsqcap_{a \in A}\codlift{F}{\bO(a)}{\tau_a} \circ \Ncodlift{T}{\bO}{\sigma}{N}
  = \codlift{F}{\bO}{\tau} \circ \Ncodlift{T}{\bO}{\sigma}{N}.$
  The first natural transformation is above $\mathrm{id}$ and the second one is above $\lambda$.
\end{proof}

\subsection{Omitted Proofs for Section~\ref{sec:examples}}

\subsubsection{Bisimilarity Pseudometric} \label{ap:bisim_pseudo}
\begin{proposition} \label{ap:prop:bisim_pseudo}
  In the setting of \S{}\ref{sec:bisim_pseudo}, the tuple $(\pfibU p \PMet \Set A \O \bO, (\times, \distFda), \sigma_\land)$
satisfies the conditions in Theorem~\ref{thm:suf_lift_dist}.
\end{proposition}
\begin{proof}
1) For each $(t, \rho), (t', \rho') \in 2 \times [0, 1]^\Sigma$, $(\sigma_\land \circ \tau_\epsilon \times \tau_\epsilon)((t, \rho), (t', \rho')) = t \land t' = (\tau_\epsilon \circ (2 \times \sigma_\land^\Sigma) \circ \lambda) ((t, \rho), (t', \rho'))$
and
$(\sigma_\land \circ \tau_a \times \tau_a)((t, \rho), (t', \rho')) = w \cdot (\rho(a) \land \rho'(a)) = (\tau_a \circ (2 \times \sigma_\land^\Sigma) \circ \lambda) ((t, \rho), (t', \rho'))$ ($a \in \Sigma$).

2) We aim to show that for each $(d_1, d_2) \in \PMet_X \times \PMet_Y$, $a \in A$, $k\colon \Ncodlift{\times}{d_{[0, 1]}}{\sigma_\land}{2}(d_1, d_2) \to d_{[0, 1]}$,
$(t_1, f_1), (t_2, f_2) \in (2 \times (X \times Y)^\Sigma)$,
\begin{align*}
&|(\tau_a \circ (2 \times k^\Sigma))(t_1, f_1) - (\tau_a \circ (2 \times k^\Sigma))(t_2, f_2)|
\\&\leq
\sup_{\substack{a' \in A,\\ k_1\colon d_1 \to d_{[0, 1]},\\ k_2\colon d_2 \to d_{[0, 1]}}}
|(\tau_{a'} \circ (2 \times (\sigma_\land \circ k_1 \times k_2)^\Sigma))(t_1, f_1)
- (\tau_{a'} \circ (2 \times (\sigma_\land \circ k_1 \times k_2)^\Sigma))(t_2, f_2)|.
\end{align*}
If $a = \epsilon$, (lhs) $= |t_1 - t_2| =$ (rhs).
If $a \in \Sigma$,
noting that $|k(x_1, y_1) - k(x_2, y_2)| \leq \max(d_1(x_1, x_2), d_2(y_1, y_2))$ for each $x_1, x_2 \in X$ and $y_1, y_2 \in Y$,
(lhs) $= w \cdot |k f_1(a) - k f_2(a)| \leq w \cdot \max(d_1(x_1, x_2), d_2(y_1, y_2))
= |(\tau_a \circ (2 \times (\sigma_\land \circ k_1 \times k_2)^\Sigma))(t_1, f_1)
- (\tau_a \circ (2 \times (\sigma_\land \circ k_1 \times k_2)^\Sigma))(t_2, f_2)| \leq$ (rhs)
where $f_1(a) = (x_1, y_1)$, $f_2(a) = (x_2, y_2)$,
$k_1(t) \coloneqq 1-d_1(x_1, t)$, $k_2(t) \coloneqq 1-d_2(x_2, t)$.

\end{proof}

\subsubsection{Similarity Pseudometric} \label{ap:sim_pseudo}
\begin{lemma} \label{ap:lem:min_max}
  $\min(a, b) - \min(c, d) \leq \max(a-c, b-d)$ for each $a, b, c, d \in [0, 1]$.
\end{lemma}
\begin{proof}
  If $a \leq b$ and $c \leq d$, (lhs) $= a - c \leq$ (rhs).
  If $a \geq b$ and $c \geq d$, (lhs) $= b - d \leq$ (rhs).
  If $a \leq b$ and $c \geq d$, (lhs) $= a - d \leq b-d \leq $(rhs).
  If $a \geq b$ and $c \leq d$, (lhs) $= b - c \leq a-c \leq $(rhs).
\end{proof}

\begin{proposition}
  $\das$ is a Lawvere metric.
\end{proposition}
\begin{proof}
  For each $a, b, c \in [0, 1]$,
  \begin{itemize}
  \item $\das(a, a) = 0$ is easy.
  \item $\das(a, c) \leq \das(a, b)+\das(b, c)$ is because:
    If $a \leq b$ and $b \leq c$, (lhs) $= c-a =$ (rhs).
    If $a \geq b$ and $b \geq c$, (lhs) $= 0 =$ (rhs).
    If $a \leq b$ and $b \geq c$, $c-a \leq b-a =$ (rhs) and $0 \leq b-a =$ (rhs).
    If $a \geq b$ and $b \leq c$, $c-a \leq c-b =$ (rhs) and $0 \leq c-b =$ (rhs).
  \end{itemize}
\end{proof}
\begin{proposition}
  \begin{enumerate}
    \item
  The codensity lifting $\codlift{F}{\das}{\tau}$ maps $(X,d)\in\QPMet$ to
  the pseudometric on $FX$ given by
    $((t_1, \rho_1), (t_2, \rho_2)) \mapsto
    max\{\das(t_1, t_2), w \cdot \max_{a \in \Sigma} d(\rho_1(a), \rho_2(a))\}$.
  \item
The binary codensity lifting
$\Ncodlift{\times}{\das}{\sigma_\land}{2}$  maps
$(d_1, d_2) \in \QPMet_X \times \QPMet_Y$ to the
Lawvere metric on $X \times Y$ given by
  $((x, y), (x', y')) \mapsto \max(d_1(x, x'), d_2(y, y'))$.
  \end{enumerate}
\end{proposition}
\begin{proof}
  \begin{enumerate}
    \item Fix arbitrary $d \in \QPMet_X$ and $(t_1, \rho_1), (t_2, \rho_2) \in 2 \times X^\Sigma$.
    It is easy to see that
    $\codlift{F}{\das}{\tau_\epsilon}(d)((t_1, \rho_1), (t_2, \rho_2)) = \das(t_1, t_2)$.
    Let us prove that
    $\codlift{F}{\das}{\tau_a}(d)((t_1, \rho_1), (t_2, \rho_2)) =
     w \cdot d(\rho_1(a), \rho_2(a))$ for each $a \in \Sigma$.
     ($\geq$): Letting $k(t) \coloneqq d(\rho_1(a), t)$,
     $k\colon d \to \das$ and
     (lhs) $\geq \das(w \cdot k \rho_1(a), w \cdot k \rho_2(a)) = w \cdot d(\rho_1(a), \rho_2(a)) =$ (rhs).
     ($\leq$): For each $k\colon d \to \das$,
     $\das(w \cdot k \rho_1(a), w \cdot k \rho_2(a)) = w \cdot \das(k \rho_1(a), k \rho_2(a)) \leq w \cdot d(\rho_1(a), \rho_2(a))$. Therefore we have (lhs) $\leq$ (rhs).
    \item Fix arbitrary $d_1, d_2 \in \QPMet_X$ and $(x, y), (x', y') \in X \times Y$.
      Let us prove $\Ncodlift{\times}{\das}{\sigma_\land}{2}(d_1, d_2)((x, y), (x', y')) = \max(d_1(x, x'), d_2(y, y'))$.
      $(\geq)$: Letting $k_1(t) \coloneqq 1-d_1(t, x')$ and $k_2(t) \coloneqq 1-d_2(t, y')$,
      $k_1\colon d_1 \to \das$ and $k_2\colon d_2 \to \das$.
      Therefore we have
      (lhs) $\geq \das(\min(k_1x, k_2y), \min(k_1x', k_2y'))
      = \max(d_1(x, x'), d_2(y, y')) =$ (rhs).
      $(\leq)$: For each $k_1 \colon d_1 \to \das$ and $k_2\colon d_2 \to \das$,
      $\das(\min(k_1 x, k_2 y), \min(k_1 x', k_2 y')) \leq \min(k_1 x', k_2 y') - \min(k_1 x, k_2 y) \leq \max(k_1 x' - k_1 x, k_2 y' - k_2 y) \leq \max(d_1(x, x'), d_2(y, y'))$ by Lemma~\ref{ap:lem:min_max}.
      Therefore we have (lhs) $\leq$ (rhs).
  \end{enumerate}
\end{proof}

\begin{proposition}
  In the setting of \S{}\ref{sec:sim_pseudo},
  the tuple
$(\pfib p \QPMet \Set A \O \bO, (\times, \distFda), \sigma_\land)$
 satisfies the conditions in Theorem~\ref{thm:suf_lift_dist}.
\end{proposition}
\begin{proof}
We can prove it in the same way as Proposition~\ref{ap:prop:bisim_pseudo}.
\end{proof}

\subsubsection{Bisimulation Metric}
\begin{proof}[Proof of Lemma~\ref{lem:pd_compat}]
  1) We can prove it by letting $d \coloneqq \Ncodlift{T}{\bO}{\sigma}{N}(d_1, d_2)$ and $z \coloneqq \sigma$ in Lemma~\ref{lem:gebler}.

  2)
  Let $P \in \PMet_X$ and $Q \in \PMet_Y$.
  $\Ncodlift{T}{\bO}{\sigma}{N}(\codlift{\mathcal{P}}{\mathbf{\Omega}}{\inf}P, \codlift{\mathcal{P}}{\mathbf{\Omega}}{\inf}Q)\sqsubseteq \lambda^{\pow *}\codlift{\mathcal{P}}{\mathbf{\Omega}}{\inf}\Ncodlift{T}{\bO}{\sigma}{N}(P, Q)$
  is equivalent to
  for any $A_1, A_2 \in \mathcal{P}X, B_1, B_2 \in \mathcal{P}Y$,
  \begin{equation} \label{eq:pd_compat}
  \sigma((\codlift{\mathcal{P}}{\mathbf{\Omega}}{\inf}P)(A_1, A_2), (\codlift{\mathcal{P}}{\mathbf{\Omega}}{\inf}Q)(B_1, B_2))\geq (\codlift{\mathcal{P}}{\mathbf{\Omega}}{\inf}\Ncodlift{T}{\bO}{\sigma}{N}(P, Q))(\lambda^\pow(A_1, B_1), \lambda^\pow(A_2, B_2)).
  \end{equation}

  If $\lambda^\pow(A_1, B_1) = \lambda^\pow(A_2, B_2) = \emptyset$,
  then
  (rhs) $= 0$ by the definition of $\codlift{\mathcal{P}}{\mathbf{\Omega}}{\inf}$ so the inequality \eqref{eq:pd_compat} holds.

  If $\lambda^\pow(A_1, B_1) = \emptyset$ and $\lambda^\pow(A_2, B_2) \neq \emptyset$,
  then
  ($A_1 = \emptyset$ or $B_1 = \emptyset$),
  $A_2 \neq \emptyset$, and $B_2 \neq \emptyset$ by definition of $\lambda^\pow$.
  Thus at least one of $(\codlift{\mathcal{P}}{\mathbf{\Omega}}{\inf}P)(A_1, A_2)$ and
  $(\codlift{\mathcal{P}}{\mathbf{\Omega}}{\inf}Q)(B_1, B_2)$ is equal to $1$ by the definition of $\codlift{\mathcal{P}}{\mathbf{\Omega}}{\inf}$.
  Because $\sigma(1, a) = \sigma(a, 1) = 1$ for any $a \in [0, 1]$,
  we have (lhs) $= 1$. It indicates the inequality \eqref{eq:pd_compat} holds.

  Assume that $\lambda^\pow(A_1, B_1) \neq \emptyset$ and $\lambda^\pow(A_2, B_2) \neq \emptyset$.
  The definition of  $\codlift{\mathcal{P}}{\mathbf{\Omega}}{\inf}$
  gives that
  \begin{align*}
  \text{(rhs)} = \max \bigl(&\sup_{(a_1, b_1) \in \lambda^\pow(A_1, B_1)} \inf_{(a_2, b_2) \in \lambda^\pow(A_2, B_2)} \Ncodlift{T}{\bO}{\sigma}{N}(P, Q)((a_1, b_1), (a_2, b_2)), \\
  &\sup_{(a_2, b_2) \in \lambda^\pow(A_2, B_2)} \inf_{(a_1, b_1) \in \lambda^\pow(A_1, B_1)} \Ncodlift{T}{\bO}{\sigma}{N}(P, Q)((a_1, b_1), (a_2, b_2)) \bigr).
  \end{align*}
  Therefore it's enough to show that
  two inequalities:
  \begin{itemize}
    \item (lhs) $\geq \sup_{(a_1, b_1) \in \lambda^\pow(A_1, B_1)} \inf_{(a_2, b_2) \in \lambda^\pow(A_2, B_2)} \Ncodlift{T}{\bO}{\sigma}{N}(P, Q)((a_1, b_1), (a_2, b_2))$,
    \item (lhs) $\geq \sup_{(a_2, b_2) \in \lambda^\pow(A_2, B_2)} \inf_{(a_1, b_1) \in \lambda^\pow(A_1, B_1)} \Ncodlift{T}{\bO}{\sigma}{N}(P, Q)((a_1, b_1), (a_2, b_2))$.
  \end{itemize}
  Let us show
  the former.
  The other one can be shown similarly.

  For any $(a_1, b_1) \in \lambda^\pow(A_1, B_1)$,
\begin{align*}
    \text{(lhs)} &= \sigma((\codlift{\mathcal{P}}{\mathbf{\Omega}}{\inf}P)(A_1, A_2), (\codlift{\mathcal{P}}{\mathbf{\Omega}}{\inf}Q)(B_1, B_2)) \\
    &\geq \sigma(\sup_{a_1' \in A_1} \inf_{a_2' \in A_2}P(a_1', a_2'), \sup_{b_1' \in B_1}\inf_{b_2' \in B_2}Q(b_1', b_2'))
    &\text{by the definition of $\codlift{\mathcal{P}}{\mathbf{\Omega}}{\inf}$} \\
    &\geq \sigma(\inf_{a_2' \in A_2}P(a_1, a_2'), \inf_{b_2' \in B_2}Q(b_1, b_2'))
    &\text{since $\sigma$ is monotone} \\
    &= \inf_{a_2' \in A_2, b_2' \in B_2}\sigma(P(a_1, a_2'), Q(b_1, b_2'))
    &\text{since $\sigma$ preserves infimums} \\
    &= \inf_{(a_2', b_2') \in \lambda^\pow(A_2, B_2)}\Ncodlift{T}{\bO}{\sigma}{N}(P, Q)((a_1, b_1), (a_2', b_2')).
  \end{align*}
\end{proof}

\subsection{Omitted Proofs for Section~\ref{sec:game}}

\begin{lemma} \label{lem:invariant_comp}
  Assume that two conditions in Theorem~\ref{thm:suf_lift_dist} hold.
A set $\mathcal{V} \subseteq |\prod_{i \in N} \mathbb{E}_{X_i}|$ is an invariant for $\duplicator$ in the composite codensity game $\compcgame{c_i}{i}{N}{T}{\sigma}$
  if and only if
  $\bigsqcup_{\vfami{P_i}{i}{N} \in \mathcal{V}}\Ncodlift{T}{\bO}{\sigma}{N}(\vfami{P_i}{i}{N}) \sqsubseteq  (\coalgmerge T\lambda{\vfami{c_i}{i}{N}})^*\codlift{F}{\bO}{\tau} \Ncodlift{T}{\bO}{\sigma}{N}(\fami{\bigsqcup_{\vfami{P_i}{i}{N} \in \mathcal{V}}P_i}{i}{N})$.
\end{lemma}
\begin{proof}
\begin{align*}
&\lambda^*\codlift{F}{\bO}{\tau} \Ncodlift{T}{\bO}{\sigma}{N}\\
&= \lambda^*\conc p\bO \circ \Sp[A]F\tau \circ \Sp[A]T\sigma \circ
  \abs{p^N}{\cptpl\bO N} &&\text{by the condition \eqref{item:approximating} in Theorem~\ref{thm:suf_lift_dist}} \\
&= \conc p\bO \circ \lambda^*\Sp[A]F\tau \circ \Sp[A]T\sigma \circ
  \abs{p^N}{\cptpl\bO N} &&\text{since $\conc p\bO$ preserves cartesian morphisms} \\
&= \conc p\bO \circ \Sp[A]T\sigma \circ \Sp[A]{F^N}{\tau^N} \circ \abs{p^N}{\cptpl\bO N} &&\text{by the condition \eqref{item:2-cell} in Theorem~\ref{thm:suf_lift_dist},} \\
& &&\text{ and discussion after Theorem~\ref{thm:suf_lift_dist}}\\
&= \conc p\bO \circ \Sp[A]{T \circ F^N}{\sigma \bul (T \circ \tau^N)} \circ \abs{p^N}{\cptpl\bO N}.
\end{align*}
From this equation,
  \begin{align*}
    &\bigsqcup_{\langle P_i \rangle_{i \in N} \in \mathcal{V}}\codlift{T}{\mathbf{\Omega}}{\sigma}(\langle P_i \rangle_{i \in N}) \sqsubseteq  (\lambda \circ T(\langle c_i \rangle_{i \in N}))^*\codlift{F}{\mathbf{\Omega}}{\tau} (\codlift{T}{\mathbf{\Omega}}{\sigma}(\langle \bigsqcup_{\langle P_i \rangle_{i \in N} \in \mathcal{V}}P_i \rangle_{i \in N})) \\
    &\text{if and only if }\\
    &\bigsqcup_{\langle P_i \rangle_{i \in N} \in \mathcal{V}}\codlift{T}{\mathbf{\Omega}}{\sigma}(\langle P_i \rangle_{i \in N}) \sqsubseteq  (T(\langle c_i \rangle_{i \in N}))^*\conc p\bO \circ \Sp[A]{T \circ F^N}{\sigma \bul (T \circ \tau^N)} \circ \abs{p^N}{\cptpl\bO N} (\langle \bigsqcup_{\langle P_i \rangle_{i \in N} \in \mathcal{V}}P_i \rangle_{i \in N})\\
      &\phantom{\text{if and only if }\bigsqcup_{\langle P_i \rangle_{i \in N} \in \mathcal{V}}\codlift{T}{\mathbf{\Omega}}{\sigma}(\langle P_i \rangle_{i \in N})} \, =\bigsqcap_{a \in A, \langle k_i \in \mathbb{E}(\bigsqcup_{\langle P_i \rangle_i \in \mathcal{V}}P_i, \mathbf{\Omega}) \rangle_{i \in N}}\big(\sigma \circ T(\langle \tau_a \circ Fpk_i \circ c_i\rangle_{i \in N}) \big)^*\mathbf{\Omega} \\
    &\text{if and only if }\\
    & \forall \langle P_i \rangle_{i \in N} \in \mathcal{V}.~\forall a \in A.~\forall \langle k_i\colon \bigsqcup_{\langle P_i' \rangle_i \in \mathcal{V}}P_i' \decent{\rightarrow}\mathbf{\Omega} \rangle_{i \in N}.~\codlift{T}{\mathbf{\Omega}}{\sigma}(\langle P_i \rangle_{i \in N}) \sqsubseteq  \big(\sigma \circ T(\langle \tau_a \circ Fk_i \circ c_i\rangle_{i \in N}) \big)^*\mathbf{\Omega} \\
    &\text{if and only if } \forall \langle P_i \rangle_{i \in N} \in \mathcal{V}.~\forall a \in A.~\forall \langle k_i \in \mathbb{C}(X_i, \Omega) \rangle_{i \in N}. \\
      &\phantom{\text{if and only if }}\bigl(\forall \langle P_i' \rangle_{i \in N} \in \mathcal{V}.~\forall i \in N~k_i\colon P_i' \decent{\rightarrow} \mathbf{\Omega} \Rightarrow \codlift{T}{\mathbf{\Omega}}{\sigma}(\langle P_i \rangle_{i \in N}) \sqsubseteq \big(\sigma \circ T(\langle \tau_a \circ Fk_i \circ c_i\rangle_{i \in N}) \big)^*\mathbf{\Omega} \bigr)\\
    &\text{if and only if } \forall \langle P_i \rangle_{i \in N} \in \mathcal{V}.~\forall a \in A.~\forall \langle k_i \in \mathbb{C}(X_i, \Omega) \rangle_{i \in N}. \\
      &\phantom{\text{if and only if }}\bigl(\codlift{T}{\mathbf{\Omega}}{\sigma}(\langle P_i \rangle_{i \in N}) \not \sqsubseteq \big(\sigma \circ T(\langle \tau_a \circ Fk_i \circ c_i\rangle_{i \in N}) \big)^*\mathbf{\Omega}  \Rightarrow\exists \langle P_i' \rangle_{i \in N} \in \mathcal{V}.~\exists i \in N~k_i\colon P_i' \decent{\not \rightarrow} \mathbf{\Omega}\bigr).
  \end{align*}
\end{proof}

\begin{lemma}
Let
$\mathcal{V} \subseteq \prod_{i \in N}
\mathrm{Obj}(\mathbb{E}_{X_i})$ be a set.
If a set
$\mathcal{V}$ is an invariant of a codensity game $\mathcal{G}_c$ for
$\duplicator$ then so is $\overline{\mathcal{V}}$.
\end{lemma}
\begin{proof}
Suppose $\mathcal{V}$ is an invariant of $\mathcal{G}_c$ for $\duplicator$.
To show that $\overline{\mathcal{V}}$ is an invariant of $\mathcal{G}_c$ for $\duplicator$,
we need to check:
for each $(a \in A, k\colon X \to \O(a))$ such that
$\tau_a \circ Fk \circ c\colon \bigsqcup_{P \in \mathcal{V}} P \decent{\not{\rightarrow}} \bO(a)$,
there exists $P' \in \mathcal{V}$ such that $k\colon P' \decent{\not{\rightarrow}} \O(a)$.
It is true because $\tau_a \circ Fk \circ c\colon \bigsqcup_{P \in \mathcal{V}} P \decent{\not{\rightarrow}} \bO(a)$ implies $\exists~P' \in \mathcal{V}.~\tau_a \circ Fk \circ c\colon P' \decent{\not{\rightarrow}} \bO(a)$ and $\mathcal{V}$ is an invariant.
\end{proof}

\begin{proof}[Proof of Proposition~\ref{prop:invariant_compositional_bisim}]
Assume that each $\mathcal{V}_i \subseteq \mathbb{E}_{X_i}$ is an invariant of $\mathcal{G}_{c_i}$ for $\duplicator$.
For any $a \in A$ and $\fami{k_i \in \mathbb{C}(X_i, \Omega)}{i}{N}$,
\begin{align} \label{eq:invariant_compositional_bisim}
  \begin{split}
  \codlift{T}{\mathbf{\Omega}}{\sigma}(\vfami{(\tau_a \circ Fk_i \circ c_i)^*\mathbf{\Omega}(a)}{i}{N})
  &= \bigsqcap_{\substack{a' \in A, \fami{k'_i\colon (\tau_a \circ Fk_i \circ c_i)^*\mathbf{\Omega} \to \mathbf{\Omega}}{i}{N}}} \big(\sigma_{a'} \circ T(\vfami{pk'_i}{i}{N}) \big)^*\mathbf{\Omega}(a') \\
  &\sqsubseteq \big(\sigma_a \circ T(\vfami{\tau_a \circ Fk_i \circ c_i}{i}{N}) \big)^*\mathbf{\Omega}(a)
  \end{split}
\end{align}
For arbitrary $\vfami{P_i }{i}{N} \in \prod_{i \in N}\mathcal{V}_i$, $a \in A$, and $\fami{k_i \in \mathbb{C}(X_i, \Omega)}{i}{N}$,
if $\codlift{T}{\mathbf{\Omega}}{\sigma}(\vfami{P_i }{i}{N}) \not \sqsubseteq \big(\sigma_a \circ T(\vfami{\tau_a \circ Fk_i \circ c_i}{i}{N}) \big)^*\mathbf{\Omega}(a)$, we have $\exists i \in N.~P_i \not \sqsubseteq (\tau_a \circ Fk_i \circ c_i)^*\mathbf{\Omega}(a)$ by \eqref{eq:invariant_compositional_bisim}.
Since $\mathcal{V}_i$ is an invariant for $\mathbf{D}$,
   $\exists P'_i \in \mathcal{V}_i.~k_i\colon P'_i \decent{\not \rightarrow} \mathbf{\Omega}(a)$.
  Thus we have $\vfami{P'_i}{i}{N} \in \mathcal{V}$ such that $\exists i \in N~k_i\colon P_i' \decent{\not \rightarrow} \mathbf{\Omega}(a)$.
\end{proof}

\begin{proof}[Proof of Theorem~\ref{thm:invariant_compositional_bisim}]
  (1) If
  $\bigsqcup_{\vfami{P_i}{i}{N} \in \mathcal{V}}\Ncodlift{T}{\bO}{\sigma}{N}(\vfami{P_i}{i}{N})$
  is a $\codlift{F}{\bO}{\tau}$-bisimulation on $\coalgmerge T\lambda{\vfami{c_i}{i}{N}}$,
  then
  $\mathcal{V}$ satisfies the inequality in Lemma~\ref{lem:invariant_comp}
  because
  $\bigsqcup_{\vfami{P_i}{i}{N} \in \mathcal{V}}\Ncodlift{T}{\bO}{\sigma}{N}(\vfami{P_i}{i}{N}) \sqsubseteq (\coalgmerge T\lambda{\vfami{c_i}{i}{N}})^*\codlift{F}{\bO}{\tau}
  \bigsqcup_{\vfami{P_i}{i}{N} \in \mathcal{V}}\Ncodlift{T}{\bO}{\sigma}{N}(\vfami{P_i}{i}{N})
  \sqsubseteq (\coalgmerge T\lambda{\vfami{c_i}{i}{N}})^*\codlift{F}{\bO}{\tau} \Ncodlift{T}{\bO}{\sigma}{N}(\fami{\bigsqcup_{\vfami{P_i}{i}{N} \in \mathcal{V}}P_i}{i}{N})$.
  Therefore $\mathcal{V}$ is an invariant by Lemma~\ref{lem:invariant_comp}.

  (2)
Since $\fami{\bigsqcup_{\vfami{P_i}{i}{N} \in \mathcal{V}}  P_i}{i}{N} \in \overline{\mathcal{V}}$ and $\codlift{T}{\mathbf{\Omega}}{\sigma}(\vfami{P_i}{i}{N}) \sqsubseteq \codlift{T}{\mathbf{\Omega}}{\sigma}(\fami{\bigsqcup_{\vfami{P_i}{i}{N} \in \mathcal{V}}  P_i}{i}{N})$ for each $\vfami{P_i}{i}{N} \in \overline{\mathcal{V}}$,
  $\bigsqcup_{\vfami{P_i}{i}{N} \in \overline{\mathcal{V}}} \codlift{T}{\mathbf{\Omega}}{\sigma}(\vfami{P_i}{i}{N})
  = \codlift{T}{\mathbf{\Omega}}{\sigma}(\fami{\bigsqcup_{\vfami{P_i}{i}{N} \in \mathcal{V}}  P_i}{i}{N})
  = \codlift{T}{\mathbf{\Omega}}{\sigma}(\fami{\bigsqcup_{\vfami{P_i}{i}{N} \in \overline{\mathcal{V}}}  P_i}{i}{N})$.
  Lemma~\ref{lem:invariant_comp} concludes the proof.
\end{proof}

 \else
\fi

\end{document}